\documentclass[%
 reprint,
 amsmath,amssymb,
 aps,
 reprint,
 amsmath,amssymb,
 aps]{revtex4-2}  
\usepackage{euscript}
\usepackage{mathtools}  
\usepackage[dvipsnames]{xcolor}
\usepackage{graphicx}
\usepackage{setspace}
\usepackage{dcolumn}
\usepackage{adjustbox}
\PassOptionsToPackage{hyphens}{url}
\usepackage{tikz}
\usetikzlibrary{automata, positioning, arrows}
\usepackage{xcolor,colortbl}
\usepackage{xurl}
\usetikzlibrary{arrows.meta}
\usepackage{bm}

\DeclareMathAlphabet{\mathpzc}{OT1}{pzc}{m}{it}
\usepackage{array}
\usepackage{makecell}
\usepackage{mathtools}

\usepackage{float}
\usepackage{mathrsfs}
\usepackage{amsthm}
\newtheorem{theorem}{Theorem}[section]

\newtheorem{lemma}[theorem]{Lemma}

\usepackage{geometry}
\makeatletter

\usepackage{csquotes}

\definecolor{purpleheart}{rgb}{0.3, 0.24, 0.64}
\definecolor{persianindigo}{rgb}{0.2, 0.07, 0.48}
\definecolor{royalazure}{rgb}{0.0, 0.22, 0.66}
\definecolor{lightpastelblue}{rgb}{0.4, 0.5, 0.95}
\definecolor{darkpastelred}{rgb}{0.56, 0.23, 0.25}
\definecolor{darkred}{rgb}{0.55, 0.0, 0.0}
\definecolor{liprose}{rgb}{0.75, 0.5, 0.51}
\definecolor{green(colorwheel)(x11green)}{rgb}{0.0, 1.0, 0.0}
\definecolor{junglegreen}{rgb}{0.16, 0.67, 0.53}

\definecolor{paleblue}{rgb}{0.69, 0.93, 0.93}

\definecolor{chestnut}{rgb}{0.8, 0.36, 0.36}

\definecolor{coralred}{rgb}{1.0, 0.25, 0.25}

\definecolor{egyptianblue}{rgb}{0.06, 0.23, 0.75}

\definecolor{jasper}{rgb}{0.88, 0.23, 0.24}
\definecolor{bluepigment}{rgb}{0.2, 0.2, 0.6}
\definecolor{burgundy}{rgb}{0.5, 0.0, 0.13}
\definecolor{colblue}{rgb}{0.2, 0.66, 0.74}
\definecolor{darkturquoise}{rgb}{0.36, 0.54, 0.95}
\definecolor{lavenderblue}{rgb}{0.5, 0.5, 0.98}
\definecolor{tractorred}{rgb}{0.80, 0.20, 0.20}

\usepackage{yfonts}

\usepackage{etoolbox}
\usepackage{hyphenat}

\newcommand\vartextvisiblespace[1][.5em]{%
  \makebox[#1]{%
    \kern.07em
    \vrule height.3ex
    \hrulefill
    \vrule height.3ex
    \kern.07em
  }
}
\usepackage{pgfornament}
\usepackage{enumitem}
\usepackage[ colorlinks=true,linkcolor=blue,urlcolor=blue,citecolor=blue,pdfusetitle]{hyperref}
\usepackage{color}
\newcommand\test[1]{%
  \language\csname l@#1\endcsname
  \parbox[t]{0pt}{\hspace{0pt}everywhere}%
}
\usepackage{microtype}

\usepackage[noline]{algorithm2e}

\SetKwFor{For}{for (}{) $\lbrace$}{$\rbrace$}

\usepackage[capitalize]{cleveref}   
\newcommand{\eq}[1]{\begin{align}#1\end{align}}
\usetikzlibrary{matrix,positioning}
\usepackage{setspace}

\usepackage{bm}
\newcommand{\tvx}{\tilde{{\pmb{x}}}}
\newcommand{\vx}{{\pmb{x}}}

\newcommand{\vu}{{\pmb{u}}}
\newcommand{\vv}{{\pmb{v}}}
\newcommand{\tvu}{\tilde{{\pmb{u}}}}
\newcommand{\tvv}{\tilde{{\pmb{v}}}}

\usepackage{listings}
\usepackage{blindtext}
\usepackage{bm}

\begin{document}
\title{Inclusive Thermodynamics of Computational Machines}

\author{Gülce Kardeş}
\email{gulcekardes@gmail.com}
\affiliation{
University of Colorado, Boulder, USA
 }

\author{David Wolpert}
\affiliation{
Santa Fe Institute, USA
} 
\altaffiliation[Also at ]{Complexity Science Hub Vienna, Austria; Arizona State University, USA; International Centre for Theoretical Physics, Italy.
}

\begin{abstract}
We introduce a framework designed to analyze the thermodynamics of an abstractly defined logical computer like a deterministic finite automaton (DFA) or a Turing machine, without specifying any extraneous parameters (like rate matrices, Hamiltonians, etc.) of a physical process that implements the computer. Earlier investigations of how to do this were based on the continuous-time Markov chain (CTMC) formulation of stochastic thermodynamics. These investigations either assumed that there was exactly zero irreversible entropy production (EP) generated by the physical system implementing the computation, or allowed the EP to be nonzero but only considered the “mismatch cost” component of the EP. In addition, they only applied to a single type of computer. Our framework neither requires that EP equal zero nor restricts attention to the mismatch cost component of EP, and is designed to apply to all types of computational machines. In contrast to earlier investigations using the CTMC-based formulation, our framework is based on the inclusive Hamiltonian formulation, in which the combination of the system of interest and the baths evolve in a Hamiltonian (or unitary) dynamics. Here, we use our framework to derive an integral fluctuation theorem for computers, in which the expectation value is strictly less than 1. We also derive an exchange fluctuation theorem, and a mismatch cost formula involving first-passage times. We analyze the EP generated by a DFA, a Markov information source, and a noisy communication channel. In particular, we use the Myhill-Nerode theorem of computer science to prove that out of all DFAs which recognize the same language, the “minimal complexity DFA” is the one with minimal EP for all dynamics and at all iterations.

\end{abstract}

\maketitle
\section{Introduction}

\label{sec:introduction}
\subsection{Background}

The thermodynamic costs of computation has been a central topic of concern for physicists and mathematicians for over a century. Early work ranges from Szilard’s analyses of Maxwell’s Demon~\cite{1929ZPhy_53_840S} to remarks by von Neumann, in which he argued that a computer operating at temperature $T$ must dissipate at least $kT\ln2$ Joule per elementary bit operation~\cite{Neumann1961JohnVN}. Landauer, Bennett, Zurek, Caves and other collaborators then built on these earlier investigations with a more extended, semi-formal analysis in the mid- to late twentieth century~\cite{Zurek1990ComplexityEA,BENNETT2003501, 10.5555/295385.295402}. 

All these early investigations were based on equilibrium thermodynamics. However, real-world computers almost always operate (extremely) far from thermodynamic equilibrium. This indicates that a more complete and detailed understanding of the thermodynamics of computation, extending beyond the analyses of the last century, must involve a formalism explicitly designed to apply to non-equilibrium systems. 

Fortunately, the last two decades have witnessed substantial advances which have extended statistical physics to include systems operating arbitrarily far from equilibrium. One of the core ideas underlying these recent
advances is to extend the definitions of thermodynamic quantities to the level of individual trajectories
of a system, i.e., to define work, heat, etc., for individual samples of the stochastic process governing dynamics of that system \cite{Sekimoto1997ComplementarityRF}. This has allowed the derivation of powerful ``fluctuation theorems'' (FTs)~\cite{Crooks_1999,JarzynskiHamiltonian, peliti2021stochastic,Esposito2010ThreeDF} that govern the probability density function of
how much work is dissipated in a process. More recent results include ``speed limit theorems" bounding how fast a thermodynamic system can change its state distribution by the amount of dissipated work it produces~\cite{Shiraishi2018SpeedLF, Funo_2019,PhysRevLett.123.110603, https://doi.org/10.48550/arxiv.2108.04261}. Similarly,  ``thermodynamic uncertainty relations" (TURs~\cite{Horowitz2019ThermodynamicUR,Liu2020ThermodynamicUR,https://doi.org/10.48550/arxiv.2101.01610,Hasegawa2019GeneralizedTU,PhysRevLett.114.158101, PhysRevLett.125.120604}) bound the  statistical precision of \textit{any} type of current within a system by the amount of dissipated work it generates. Other recent results include various bounds relating dissipated work to stopping times and first-passage times~\cite{PhysRevX.7.011019,Manzano2021ThermodynamicsOG,Gingrich2017FundamentalBO,PhysRevLett.125.120604,Garrahan_2017}. 

This new field is called 
stochastic thermodynamics (ST), and has two main 
approaches. Most of the research in ST has been based on
considering systems of interest (SOIs)
that are coupled to one or more infinite external reservoirs (also called baths). In this standard approach, the reservoirs are assumed to always be
at thermal equilibrium, e.g., due to separation of timescales or coarse-graining~\cite{Seifert2012StochasticTF,VANDENBROECK20156,peliti2021stochastic,PhysRevX.7.021003}. Therefore there is no dynamic model of the reservoirs, and typically only an indirect model of the coupling of the reservoirs to the SOI, via conditions on the allowed dynamics of the SOI. 
In this approach the SOI itself evolves according to
a continuous-time Markov chain (CTMC) \cite{VANDENBROECK20156}. 

As described below, one of the central results in CTMC-based 
stochastic thermodynamics is a formula for the time-derivative
of the Shannon entropy of distribution $p_t$ over
the states of the SOI at time $t$:
\begin{equation}
\dfrac{dS(p_t)}{dt} = \dot{Q}(p_t) + \dot{\Sigma}(p_t)
\end{equation}
The first term on the RHS is called the ``entropy flow'' rate (EF).
In many settings it can be identified with the rate of heat exchanged between the SOI and the reservoirs. The second term
is called the ``entropy production'' rate (EP). Crucially, it is never negative. (In many scenarios the second law of thermodynamics is a consequence of this non-negativity.) The physical process of the system evolving according to the CTMC is thermodynamically reversible iff the EP rate equals $0$. In general, that can only occur if the process is proceeding semi-statically slowly \cite{VANDENBROECK20156}. 

Another substantial portion of ST research instead adopts an ``inclusive Hamiltonian framework". In this
approach the external reservoirs can be finite or infinite, but they have a finite number of
degrees of freedom. Moreover, it is not assumed that they are always at thermal equilibrium. Instead, a deterministic invertible dynamics is defined over the full physical system, including both the SOI and
the external reservoirs. Often, it is also assumed
that the initial distribution over the joint SOI-reservoirs is a product distribution, i.e., that the SOI and the reservoirs are initialized in statistically independent processes~\cite{JarzynskiHamiltonian,PhysRevLett.122.150603,massimiliano_lindenberg_vandenbroeck,PhysRevLett.123.090604,PhysRevLett.92.230602, PhysRevLett.116.020601, strasberg2015thermodynamics, Talkner2020CommentO}. 

Importantly, this Hamiltonian framework is 
not only applicable to classical systems. It is also
one of the common ways  to model open quantum systems. In
those quantum models, one
explicitly specifies a unitary operator governing the joint dynamics of an SOI together with the
external systems coupled to the SOI, with partial traces used to evaluate the dynamics of the SOI by itself~\cite{Nielsen2000QuantumCA, Breuer2002TheTO}. This model of quantum systems is central to recent research on quantum information processing~\cite{Nielsen2000QuantumCA}. 

Our goal in this paper is to build on this previous work in ST, to construct a formalism 
for analyzing the thermodynamics of arbitrary computational systems that depends solely on
the logical dynamics of those systems, without further specifying any of the low-level details of the physical process
that implements that machine. We want to be able to analyze the thermodynamics of \textit{just the dynamics of the computational machine}, with our conclusions
not changing based on extraneous physical parameters that are not fixed by the dynamics of the computational machine. 

The CTMC-based approach to ST has made some progress
towards this goal. Previous 
research in this category has fallen into two classes:
\begin{enumerate}
\item In the first class,
it is assumed that EP = 0 \cite{Wolpert_2020, Kolchinsky2019ThermodynamicCO, thermo_info}. As mentioned above,
in general this
restricts us to considering systems that are evolving 
infinitesimally slowly. In this situation, the total
EF --- the total heat exchange with the reservoirs --- is the
change in the Shannon entropy. This allows us to derive expressions for the EF
by considering only the logical computation implemented by the
system, together with the initial distribution over
its states, without considering any extraneous
physical parameters; 
\item The second class of investigations
focuses on the EP, ignoring the EF \cite{Kolchinsky2019ThermodynamicCO, https://doi.org/10.48550/arxiv.2208.06895}. Specifically,
these papers focus on what is called the ``mismatch cost'' contribution to the EP. This contribution to the EP
is always non-negative. Like the EF in systems with zero EP,
this contribution to the EP is determined fully by the 
logical computation 
that the system implements, together with the 
initial distribution over
its states, without any dependence on extraneous
physical parameters.
Moreover, this contribution to EP is nonzero 
if we assume the physical system implementing the
computation evolves periodically, with one iteration of
the computation performed in each period. This
is almost always the case in real-world computers.
\end{enumerate}

These analyses have revealed that there are unavoidable trade-offs among the
thermodynamic resources used in 
specific nonequilibrium physical systems,
in particular systems that perform computation. Some of the trade-offs uncovered relate to the speed of a computation, its noise level, and whether the computational system is thermodynamically ``tailored'' to perform a given computational task~\cite{DavidWolpertSTComputation, Wolpert_2020, Kolchinsky2019ThermodynamicCO,brittain2021build,strasberg2015thermodynamics,Freitas2021StochasticTO,Gao2021PrinciplesOL,e21030284}. 

\begin{figure*}[ht]
\centering
   \includegraphics[width=0.919\linewidth]{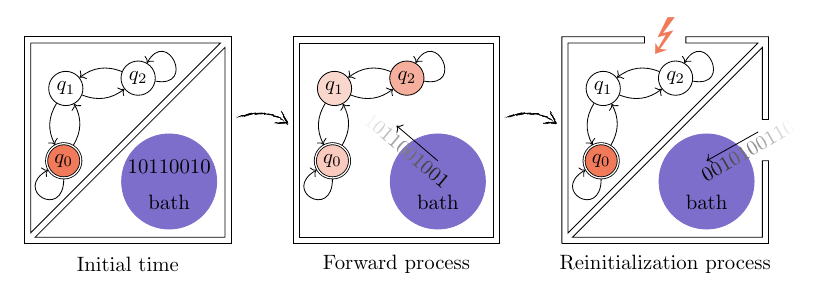} 
   \caption{Schematic depiction of a full thermodynamic cycle of a DFA,  special type of logical computer. The
   cycle starts with
   the DFA in its initial state and the environment's state --- a bit string --- set to a random sample of an associated distribution, modeled
   as a Boltzmann distribution for 
   a bath Hamiltonian. The forward process of the DFA occurs
   as the DFA reads in the input string generated by the  environment. Once the DFA halts, an infinite external reservoir is connected to the DFA, and a second, independent one is connected to the environment. These external reservoirs are used to reinitialize the DFA and
   environment, respectively, in preparation for the next run of the DFA.} 
     \label{fig:hamiltonian}
\end{figure*}

In this paper we show that by adapting the 
Hamiltonian framework, we can analyze
the energetic costs associated with a computational task 
\textit{in toto}, without specifying any extraneous physical parameters
that are not already given in the computer science (CS) theory definition
of the computation. In particular, by using
this version of the Hamiltonian framework
we can avoid the need for making restrictive assumptions 
on the speed of the process (in contrast to (1)) or 
for ignoring components of the energetic
cost (as in (2)). In addition, our
new framework applies to arbitrary computational machines.
Moreover, it is based on the assumption of the 
Hamiltonian framework that the the states of SOI and the
of the bath(s) are independent of one another when
the process starts. This assumption is perfectly suited to analyzing computational machines that 
receive external inputs after they have been initialized --- which is the case in all CS theory.

As a result of these attributes,
the framework we introduce opens the possibility of analyzing the trade-offs among all of the thermodynamic resources that are consumed in any computation, in the  sense of the term
meant in computer science involving 
abstract logical 
variables, rather than 
focus on the thermodynamic resources consumed in
a specific physical
system that implements some specific computation.

Trade-offs among the minimal amounts of resources required to perform a given computation
have not only been considered in ST.
Indeed, such trade-off are a central concern of CS theory. However, traditional CS theory has considered trade-offs among quantities different from those considered in the thermodynamics of computation \cite{Aaronson2005NPcompletePA}. For instance,
one of the most important examples of a trade-off considered in CS theory is the
relation between the amount of memory needed by a computational machine 
to perform a given computation and the number of iterations required to perform that computation \cite{10.5555/1540612}. 

Despite this parallel between the interests of ST 
and CS theory in the trade-offs among
the resource costs involved in computation, very little research has been
done on how the resource costs investigated in ST are related to the resource costs so central to CS. In this paper, after introducing our framework, 
we use it to start to lay the foundations
for investigating that relationship.

\subsection{Reinitialization entropy production}
\label{sec:physics_to_be_cs}

Our starting point is to note that many models of computational systems can be decomposed into two or more interacting subsystems \cite{BOOK1998Elements}. The first is the computational machine itself. Examples of such machines range from deterministic finite automata (DFAs)
to Turing machines (TMs) to concurrent processes. In addition to the computational machine, there are always one or more external processes which provide inputs to and receive outputs from the computational machine.

We physically ground this decomposition by supposing that there are two types of degrees of freedoms
in physical devices that implement computational machines \footnote{In the real world
there will be many
more types of degree of freedom of the full computational system than
just ``accessible'' and ``inaccessible'', i.e., different degrees of
freedom will lead to different \textit{kinds} of thermodynamic interpretations. Here, for simplicity, we assume just the fundamental two.}. First, the {\textbf{accessible}} degrees of freedom are those that arise
directly in the specification of the computational machine executing a well-defined computational task. As an example, in a DFA, one could
take the accessible degrees of freedom to be the computational state of the DFA. Whatever our
computational machine is, we suppose that
there is an engineer who builds a \textbf{physical device}
that implements that computational
machine. More precisely, we suppose that they
build a physical device some of whose physical
variables correspond to the
accessible degrees of freedom of the computational machine. So the dynamics of the physical device implements the computational task across the accessible degrees of freedom.

We refer to the dynamics of the full physical system
as it implements a given computation as
the \textbf{forward process}.
In real world scenarios, it is expected that the engineer will use the physical
device repeatedly. So after every use of the device in the forward process,
the device needs to be reinitialized for its next use, as illustrated in \cref{fig:hamiltonian}. 
We suppose that the engineer has complete freedom
to design how the physical device controls the accessible degrees of freedom during that
re-initialization (hence the name ``accessible'').

The dynamics of the physical device is driven by the external world which provides
stochastic inputs to the device. Supposing that the distribution of those inputs is fully specified, 
at all times there is a well-defined probability
distribution over the accessible degrees of freedom in the physical device. We suppose that the engineer knows exactly what that probability distribution is at the time that the device has completed a computation, and can use that knowledge to design a thermodynamically efficient physical process to re-initialize those accessible degrees of freedom. 

Formally, as described in the following sections, these properties mean that we can lower-bound
the amount of work the engineer will need to \textit{spend} to re-initialize the 
(accessible degrees of freedom in the)
computational machine. This lower bound
is given
by how the non-equilibrium free energy over those degrees of freedom changes, when the
engineer transforms the ending distribution over the
accessible degrees of freedom back to the initial distribution.

In addition to the accessible degrees of freedom, there is also a set of 
degrees of freedom that are \textbf{inaccessible} to the engineer, that comprise the external world which interacts with the computational machine. For example, in the case of a DFA, one could consider the entire string of inputs read into the DFA 
as being inaccessible. We suppose that the engineer knows the distribution
over the inaccessible degrees of freedom when the computation halts, just like they
know the distribution over the accessible degrees of freedom. However, we also suppose that
the engineer has no control over how the inaccessible degrees of freedom are re-initialized.
Formally, as described in the following sections, this means that we can upper-bound
the amount of work the engineer can \textit{extract} when the 
inaccessible degrees of freedom are reinitialized.
This upper bound is given by the amount of heat that would
be produced if those degrees of freedom were re-initialized in an uncontrolled
manner, i.e., by coupling them to an idealized, infinite heat bath whose Boltzmann
distribution is the initial distribution over the inaccessible degrees of freedom~\footnote{In many real world scenarios, the engineer in fact cannot extract \textit{any} work from the re-initialization of the inaccessible degrees of freedom. Here, we are simply
stipulating that the best they could possibly do,
in any model of any computational machine, is extract this heat transferred
in from an infinite, external bath
--- that is essentially our definition of inaccessible degrees of freedom.}. 

The difference between this minimal amount of work that needs to be spent
(to re-initialize the accessible degrees of freedom) and the maximal amount
of work that can be extracted (by re-initializing the inaccessible degrees of
freedom) combine to provide a lower bound on the amount of work that will be \textit{dissipated}
--- irretrievably lost --- every time the computational machine is run. This bound applies independent
of the details of the actual physical system that implements
the computational machine, since by the second law of thermodynamics, 
those details can increase the total amount of dissipated work, but cannot reduce it. 

Importantly, this bound on the expected
disssipated work of the reinitialization 
process exactly equals the expected entropy
production of the forward process~\cite{massimiliano_lindenberg_vandenbroeck}.
Accordingly, we will refer to this bound as
the expected \textbf{reinitialization entropy production} (REP) of the system. 
It is important to emphasize that the REP
depends only on the thermodynamics of re-initializing
the machine after it has completed a computational task. It does not reflect any extra 
dissipated work that arises in the forward process,
while that computation runs. 

From now on we will refer to 
the physical variables that we suppose the engineer is directly interested in, and that are accessible to them, as
the \textbf{system of interest} (SOI). (So these are 
the variables in the physical device that the engineer
is directly interested in.)
All variables not in the physical device are considered to be inaccessible. These are partitioned into one or more \textbf{baths} (or \textbf{reservoirs}). (As illustrated below, the baths will implement the sequence
of inputs into and outputs from the computer.)

Sometimes we will need to distinguish between the abstract computer together with its abstract sequence of inputs, as considered in CS, and the physical system that implements 
that computer together with
its inputs. In such cases we refer to the former as the \textbf{logical computer}, and refer to the
latter as the \textbf{physical computer}. So for example,
the physical computer comprises the SOI and the set
of all the baths.

A key part of our framework is a coordinate transformation between states of 
the logical computer and those of the physical computer.
However, when care to distinguish those two types of computer
is not needed, we will sometimes
just use the term \textbf{computational system},
implicitly relying on context to determine whether we mean the logical computer or the physical computer.
We will also
sometime use the term \textbf{computational machine} to refer
to that part of the computational system that does not
involve the inputs, again relying on context to
determine whether we mean the SOI or 
the abstract, mathematical computer that the
SOI implements.

\subsection{The forward processes}
\label{sec:forward_process_intro}

To investigate the thermodynamics of the forward process, we need to specify the initial distribution over the set of all the variables in the computational system, both those in the SOI and those in the baths. Often this step is skipped in CS, since that distribution is not relevant to the questions being investigated. However, this step is crucial for us. Indeed, when combined with the dynamics of the full computational system, that initial distribution fixes the final distribution --- and it is the relation of those two distributions that determines the value of the EP.

Typically in computational models, if the initial distribution over the full computational system is in fact specified,
it is a product distribution over the SOI and the bath(s). 
Concretely, it is almost always the case in those analyses that the initial state of the computational machine is statistically independent of the initial state of the external world generating the stream of inputs. 
Accordingly, in our analysis we presume that the joint physical system is initially
in a product distribution.

Next, we need to specify the dynamics of the physical computer during the forward process, starting from such a product distribution. We suppose that the
dynamics during the forward process is logically reversible (and therefore deterministic).
There are three reasons for this:
\begin{itemize}
    \item{Define the ``physical computer EP'' (PEP) to be the EP that would be generated by a physical computer
    that completes a given computation. 
   There are infinitely many Markov processes that
   could be used for that physical computer
   and that are thermodynamically reversible, resulting in zero PEP. 
   In particular, 
a deterministic and reversible process is such a Markov process. 
These processes have a special property though: by the mismatch cost theorems \cite{JStatMech}, any forward process that generates zero PEP for one initial distribution over the joint system will also generate zero PEP if run  with a different initial distribution \textit{if that process is reversible and deterministic}. In 
	contrast, if the forward process is non-invertible, then it is possible to implement it with zero PEP only for one specific initial distribution. Moreover,
 the mismatch cost theorem also establishes that any other initial distribution for such a non-invertible dynamics will generate strictly positive PEP, no matter what the precise form of that (non-invertible)  dynamics. 

In other words, if the forward process is non-invertible, then in general there must be nonzero dissipated work if one does not initialize the
	   full system with the unique initial ``prior'' distribution of the (full system)
	   physical process. Moreover, the precise amount of that dissipated work will depend both on the actual distribution and on the prior distribution. Accordingly, one cannot calculate PEP without specifying that prior distribution for such a stochastic process. 
    
    The prior distribution is in turn
	   specified by the precise details of the physical computer, details that have nothing to do with they dynamics of the logical computer. The result is that for such non-invertible stochastic processes, we cannot calculate the PEP without specifying some of those precise details of the physical computer which are absent in the associated logical computer.
	   In contrast,  
    a physical computer that has invertible dynamics 
    can generate zero PEP \textit{no matter what the initial distribution and prior distributions are}. 
    
    This means we can ignore the issue of what the prior distribution is --- so long as the dynamics is invertible. So by using an invertible dynamics we can focus on the thermodynamics arising from just the logical
    computer, without concern for the parameters of the underlying physical process that implements that computer.
	   }

\item Many of the machines considered in CS theory are deterministic, and many are stochastic. We want our framework to be able to represent systems with either kind of dynamics in a straightforward way.
As we show below,
this can be done if we restrict attention to systems with deterministic dynamics. In particular, it is straightforward to implement an arbitrary stochastic evolution of the SOI using a fully deterministic and invertible joint system.
(This can be done by using the random initialization of the baths to introduce stochasticity into the dynamics
of the SOI as and when needed.)

\item Another advantage of our using deterministic, invertible dynamics of a system is that if the \textit{full}
system, including the baths, has only a finite number of degrees of freedom, then we are in precisely the setting of the inclusive Hamiltonian framework, mentioned above. 
This means in particular that our results should 
carry over with minor modifications to the quantum thermodynamics of open quantum
systems. (In contrast, there are nontrivial difficulties in
formulating the dynamics of the SOI in
open quantum systems in terms of CTMCs, which substantially restricts our ability to analyze such
systems
using the CTMC-based version of stochastic thermodynamics.)

Accordingly, to 
complement our analysis of the REP, below we also present some new results concerning the PEP as defined under the Hamiltonian framework of the forward process; we refer to this quantity as the HEP (Hamiltonian framework EP). 
\end{itemize}

First, we show below that the expected value of the HEP equals the expected value of the REP. 
Next, we derive an integral fluctuation theorem (IFT) and an exchange FT (XFT) for the HEP. As a final contribution to understanding of the HEP, we confirm that the mismatch cost formula holds for the HEP~\footnote{As
an aside, note that even though deterministic invertible dynamics
for the full computational system is typically used to motivate the \textit{formula} for the HEP, given that formula, the actual dynamics during the forward process has no effect on the expected value
of the HEP; that expected value
is fixed by the initial and final (pre-reinitialization) distributions of the joint system,
no matter how that final distribution is generated
from the initial distribution.}.

\subsection{Results and roadmap}
\label{sec:roadmap}

We refer
to this minimal model of the thermodynamics of a computational machine and its
external environment during a forward process
followed by a reinitialization process
as the \textbf{inclusive
thermodynamics of computational machines}.

At a high level, we have two sets of results
concerning inclusive thermodynamics:
\begin{enumerate}
\item 
Some of our results concern the 
the thermodynamics of the forward process, 
considered from the perspective of the Hamiltonian framework, 
adapted to computational machines. 
We emphasize that these results are actually
more general than our analysis of computational
machines, as they apply to any use of the
Hamiltonian framework. Specifically, we have derived the mismatch cost formula for the Hamiltonian framework, and also an IFT and an XFT within this framework.

\item  One might be concerned about the applicability of the Hamiltonian framework particularly
to systems like DFAs, since that framework would only apply if we could assume that
input strings to the DFA were generated by repeatedly sampling a Boltzmann distribution:
in the real world, input strings are generated by engineers. Accordingly, we consider a lower bound on the expected dissipated work of the 
reinitialization process. We show that the lower bound derived this way –which is the REP– also lower bounds the EP of the Hamiltonian framework.
Hence, all our results concerning a lower bound on the
expected dissipated work in the reinitialization process of an inclusive model of a
computational machine also apply to the HEP of
that machine. 
\end{enumerate}

More specifically, in this paper, 
we employ our framework to investigate the dissipation costs of three computational systems: 
\begin{enumerate}
    \item 
DFA, which is a foundational model of computation that underlies more general models such as the TM, \item Markov sources of information theory, 
\item Communication channels central to the theory of communication.
\end{enumerate}

Our paper is organized as follows: In \Cref{sec:mathematical_framework}, we provide the elementary concepts of the inclusive framework. In \cref{sec:preliminaries}, we formally define the three different computational systems we analyze in this paper. In \Cref{sec:inclusive_formulation_DFAs_and_information_sources}, we present the mathematical basis for our framework. In \cref{sec:unilateral_decomposition_dynamics}, we  develop the inclusive thermodynamics of DFAs and derive the lower bound on dissipated work as REP. In \cref{sec:thermodynamic_complexity.information_source}, we use a special class of DFAs to model Markov information sources. In \cref{sec:multilateral_decomposition_dynamics}, we extend \cref{sec:thermodynamic_complexity.information_source} to model communication channels. Subsequently in \cref{sec:rate_distortion_channel}, we formulate the rate-distortion problem in the theory of communication using the inclusive thermodynamic quantities.

In \cref{sec:integral_ft_machines}, \cref{sec:prior_costs_dfas}, and \cref{sec:xft_communication} we present our results concerning the thermodynamics of the forward processes of physical computers which implement DFAs. In \cref{sec:integral_ft_machines} we derive an 
IFT for the HEP. In contrast to conventional
IFTs though, here we find that the expectation of
(the exponential of negative of the) HEP is upper-bounded by $1$, rather than equal $1$ exactly, as it is the case
in conventional IFTs.
Intuitively, this is because typically the
logical computers considered in CS theory
have a single, unique initial state. So the
distribution over the states of the SOI at $t=0$
is a delta function. This means that any 
reverse trajectory that does not end in that
initial state does not contribute to the IFT calculation. In other words, the IFT
we derive equals $1$ minus the probability
of such an impossible reverse trajectory.
In \cref{sec:prior_costs_dfas}, we derive a mismatch cost result concerning the PEP of the Hamiltonian framework. We derive our mismatch cost result with respect to marginal distribution over the states of the SOI, so it differs from the mismatch cost
of the full system discussed in \cref{sec:forward_process_intro}. \footnote{As emphasized in \cref{sec:forward_process_intro}, full system mismatch cost is independent of the actual initial distribution since the full system dynamics is deterministic and invertible.} Next in \cref{sec:xft_communication}, we extend a previously derived XFT for the HEP to scenarios with multiple baths.

In the remaining sections, we analyze the connections between a CS measure of complexity over DFAs and the EP of executing DFAs. We prove in \cref{sec:thermo_complexity_theorem} that for equivalent DFAs with different size complexities, executing a minimal complexity DFA results in the minimal EP at all iterations.

We conclude with a discussion,
where we describe methods and features essential to both CS and inclusive thermodynamics. We describe a few research directions where our framework might prove fruitful. 

As we will later come back in \cref{sec:discussion}, we are interested in implementing our framework to analyze many computational systems, ranging from push-down automata to TMs. In this work, we mainly focus on introducing the framework and illustrating it through computational systems which can implement DFAs.

\section{General framework: Inclusive formulations of computational systems}

\label{sec:mathematical_framework}
We will consider physical systems which evolve in
discrete time. These systems have at least two
components: an SOI,
and a set of one or more external environments which interact with the SOI, and are referred to as baths. 
The discrete time physical forward processes will evolve the joint system of the SOI and the bath(s)
from an initial joint distribution, until some
ending condition is reached (i.e., until computational task is implemented fully). After that
the joint distribution is reinitialized, to start 
a next physical forward process, where another computational task is implemented. We use the term \textbf{computational cycle} to mean such an entire process, taking the physical system from one
initialized distribution to the next.

As an example, much of our analysis below
concerns DFAs, a special type of logical
computer. Loosely speaking, a DFA is a system with a 
finite state space $S$, having elements
$s$. Initially the DFA is initialized to a special ``start state''. After that it 
receives a sequence of exogeneously generated symbols
called a ``string".
Those symbols are all elements of a finite alphabet, $\Sigma$ (e.g., the binary alphabet $\Sigma
_{0,1}$).
As the DFA iteratively
receives those symbols it makes associated
transitions among its possible states. The state of the
DFA when a termination
condition is reached (e.g., when the string ends) defines the
computation that the DFA performs on that string
it received. After it performs
such a computation, the DFA is reinitialized in its start state.  (See \cref{sec:preliminaries} for the formal definition of a DFA.)

We write a generic symbol from $\Sigma$ as $y$.
In this paper, we suppose that $\Sigma$ includes a special blank symbol ‘‘$\varepsilon$’’. As usual, we denote the set of finite strings of symbols from the alphabet as $\Sigma^{*}$ \cite{BOOK1998Elements}. We write a string that a DFA receives as $\omega$, having
length $|\omega|$. In our simplest physical model of the DFA, we identify the state of the DFA as the
state of the SOI, and the string $\omega$
with the state of the bath.  In this version we consider below, the state of the bath does not change
during a computational cycle. This has two consequences.
First, it means that we identify the initial distribution of
the state of the bath with the distribution over strings
that will be received by the DFA. Moreover, it means
we must augment the state space of the SOI. In addition to specifying
the state of the DFA, the state of the SOI must specify
an integer-valued pointer, $z$, to keep track of which element of the string $\omega$ is the current one. 
In other words, $z$ gives the iteration time $t$ at which the symbol $\omega[z = t] = y $ is received by the DFA. 
For later convenience, we define $\omega[-z]$ 
to mean the string of symbols of $\omega$ that are \textit{not} read at iteration $t = z$, $\omega[t=1]\dots\omega[t=z-1]\omega[t=z+1]\dots\omega[t=|\omega|]$.

In this simple implementation of a DFA, the state space of the full computational system is the set of all triples $(s, z, \omega)$. 
For the reasons given above, we suppose that
the dynamics governing this full space is deterministic. We will let strings $\omega$ be sampled at the beginning of each computational cycle, so that $\omega$ is time independent, and it is always possible to reconstruct the past history of a DFA's states. As a result, any distribution over the full state evolves by permuting which state has which probability, but doesn't actually change the multiset of the probability values of all joint states. Since entropy is a unique
function of that multiset of probability values, this means that the entropy over the
full state space is constant in time.

In order to investigate the associated thermodynamics, we must explicitly decompose
the full state space into a Cartesian product of two spaces: the set of states
of the SOI and the set of states of the bath. The states of the SOI
contain the accessible degrees of freedom, whereas the states of the bath are those inaccessible degrees of freedom.

As pointed out in \cref{sec:physics_to_be_cs}, we assume that the state of the SOI is ``accessible" to the engineer
once the DFA has completed a run. This implies that the engineer is allowed to reinitialize the initial distribution of SOI states by implementing any desired external work protocol over the SOI, while the SOI
is coupled to an infinite external thermal reservoir at temperature $1 / k_B$. In addition, as is conventional in the ST literature
of information processing, we assume that the Hamiltonian of the SOI is uniform at both the beginning and end
of any run, with the
same value at those two times \cite{thermo_info}. More precisely, we say that there is a constant $\mathrm{A}$ such that both the Hamiltonian of the SOI at the beginning of
the run and at end of the run 
have the value $\mathrm{A}$, independent of the state of the SOI. 

Based on these two assumptions, we can exploit the generalized Landauer bound of modern ST:  
the minimal free energy needed for re-initializating an SOI is given by the change in the
entropy (of the distribution over states) 
between the SOI's ending distribution
and its re-initialized distribution
\footnote{This use of ST implicitly assumes that the re-initialization is done via a CTMC, whereas we are careful \textit{not} to assume that the DFA itself evolves in a Markov process.}.

For the reservoirs, we assume that their states are ``inaccessible" to the engineer once the DFA has completed a run. Suppose that at the end of the computational process, the distribution over states of the bath is re-initialized,
just like the distribution over states of the SOI. Since the states of the bath are inaccessible, 
the engineer will not be able to implement this
reinitialization in a thermodynamically optimal
manner.
We suppose that the maximal free energy that can
be extracted in this reinitialization of the bath distribution would
occur by the bath's
relaxing to the thermal equilibrium of a fixed Hamiltonian, while being coupled to 
to an infinite external thermal reservoir at temperature $1 / k_B$. This Hamiltonian
is chosen so that the associated Boltzmann distribution is
the desired initial distribution of bath states at the beginning of the subsequent run. As opposed to
the re-initialization of the SOI distribution, the engineer is not allowed to implement \textit{any}
external work protocol acting on the bath distribution as it is
reinitialized. 

\section{Preliminaries}
\label{sec:preliminaries}
In this section we provide the definitions of DFAs, Markov information sources, and communication channels, as they are used in our paper. 

\subsection{Basic concepts and terminology}
\subsubsection{Regular languages and finite automata}

A \textbf{deterministic finite automaton} (DFA) is a five-tuple $M=\left(S, \Sigma, f, q_{0}, K\right)$, where $S$ is a finite set of states, $\Sigma$ is a finite alphabet of observable symbols, $f: S \times \Sigma \rightarrow S$ is a transition function mapping a current input symbol and the current state to another state, $q_0$ is a start state, and $K$ is a set of accepting states.
An input string $\omega$, i.e. a sequence of symbols from $\Sigma$, is accepted by a DFA if the last state entered by the machine on that input string is in $K$. A \textbf{language} recognized by a DFA is the set of strings that it accepts, $L(M)=\{\omega \in \Sigma^{*}\mid {f}\left(q_{0}, w\right)$ $\in$ $\left.K\right\}$. Equivalently, the language of a DFA is the decision problem it solves \footnote{A DFA gives a Boolean answer on any input string by answering True if the state after reading the string is an accept state and by answering False otherwise.}. $L$ is a \textbf{regular language} if there is a DFA $M$ which recognizes $L$. 

Given a language $L \subseteq \Sigma^{*}$, a pair of strings $a, b  \in \Sigma^{*}$ are equivalent with respect to $L$, i.e. $a \sim_{L} b$, if for all $w \in \Sigma^{*}$ we have
$a w \in L$ if and only if $b w \in L$. $\sim_{L}$ is an equivalence relation over $L$. For each string $a$, its \textbf{equivalence class} $[a]$ is the set of strings equivalent to it, $[a]=\left\{b \in \Sigma^{*}\mid a \sim_{L} b\right\}$. A \textbf{minimal DFA} that recognizes a regular language $L$ is a DFA which has one state for each equivalence class $[a]$. For any regular language $L$, the Myhill-Nerode theorem (MN) says that the minimal DFA $M_{\text{min}}$ in the set $\Omega(L)$ of all possible DFAs that recognize $L$ is unique up to relabeling of the DFA states (See \cref{fig:fig2_21233} and \cref{fig:fig2_3}). Key concepts and the generic proofs of the MN can be found in \cite{hopcroft_book, BOOK1998Elements}.

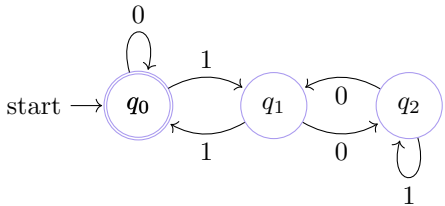
\begin{figure}[ht]
\centering
\begin{tikzpicture}[shorten >=1pt,node distance=1.8cm,on grid,auto,every state/.style = {draw = purple!80!red!20!blue!40!}]
\node[state, initial] (q0) {$q_0$};
\node[state, accepting] (q0) {$q_0$};
\node[state, right of=q0] (q1) {$q_1$};
\node[state, right of=q1] (q2) {$q_2$};
\path[->] 
    (q0) edge  [loop above] node {0} (q0)
          edge [bend left=30] node {1} (q1)
    (q1) edge  [bend left=30] node {1} (q0)
          edge  [bend right=30]  node [swap] {0} (q2)
    (q2) edge [bend right=30] node {0} (q1) 
          edge [loop below] node {1} (q2);
\end{tikzpicture}
\caption{The minimal DFA $M_\text{min}$ which recognizes the language $L$ of strings divisible by $3$.}
\label{fig:fig2_21233}
\end{figure}

\begin{figure}[ht]
\centering
\begin{tikzpicture}[shorten >=1pt,node distance=1.5cm,on grid,auto, every state/.style = {draw = purple!40!red!80!blue!20!}]
\node[state, initial] (q0) {$q_0$};
\node[state, accepting] (q0) {$q_0$};
\node[state, right of=q0] (q1) {$q_1$};
\node[state, right of=q1] (q2) {$q_2$};
\node[state, right of=q2] (q3) {$q_3$};
\path[->] 
    (q0) edge  [loop above] node {0} (q0)
          edge [bend left=30] node {1} (q1)
    (q1) edge  [bend left=30] node {1} (q0)
          edge  [bend right=30]  node [swap] {0} (q2)
    (q2) edge [bend right=30] node {0} (q1) 
          edge [bend right=30] node {1} (q3)
    (q3) edge [loop above]  node {1} (q3) 
          edge [bend right=60] node {0} (q1);
\end{tikzpicture}
\caption{A non-minimal DFA $M'$ which recognizes the language $L$ of strings divisible by $3$. $M'$ can be minimized by partitioning its states into groups of states that are mutually indistinguishable with respect to $\sim_{L}$. If we replace each such group with a single state, we get an equivalent DFA with minimum number of states, isomorphic to $M_\text{min}$ of \cref{fig:fig2_21233}.} \label{fig:fig2_3}
\end{figure}
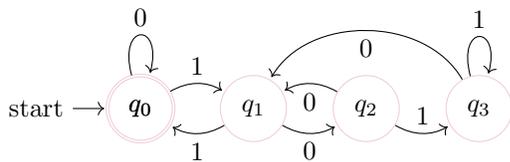

The implementations of regular languages (as decision problems solved by DFAs) concern CS problems of computational complexity. There are two central complexity issues of interest: time and space complexity, and descriptional complexity \cite{Brzozowski_Review_article}. The focus of this article is on the size complexity, which is a primary definition of descriptional complexity for regular languages and DFAs \cite{hopcroft_book, HOLZER2011456}. In accordance with the CS conventions, we define the size complexity of a given DFA as the number of states of that DFA.

Formally, DFAs can process both finite length and semi-infinite length strings (or even bi-infinite strings, as considered for example in the literature on sofic shifts \cite{Kitchens1997SymbolicDO}). Moreover, if the length of the string is finite, in general the precise length is allowed to vary.

In our main text though, for simplicity, we consider scenarios where the computational system runs for $\tau$ iterations before it is reinitialized. Accordingly the iteration $t$ will take values in the
set $\mathcal{T} := \{0, 1, 2, \ldots, \tau\}$.
To capture the possibility that the input string might have length $n < \tau$,
 we simply set all symbols $\omega[i]$ for $n < i \le \tau$ to the blank symbol, which does not occur anywhere earlier in $\omega$. We then require that no matter what state the DFA is in, if the input symbol is a blank, it stays where it is.
(See \cref{sec:app_1} and \cref{sec:app_2}  for discussion on how to extend our thermodynamic analysis, provided in \cref{sec:thermo_inclusive}, to semi-infinite and bi-infinite length strings, respectively.)

\subsubsection{Information sources and communication channels}

DFAs serve as the core information-processing system of many logical computers. In particular, Markov information sources are DFAs with the modification that rather than \textit{read in} random symbols and change states accordingly, they 
make stochastic state transitions and \textit{emit} symbols accordingly. If noise gets added
to the symbols emitted by such a Markov source, the result
is a noisy communication channel. 

More precisely, an \textbf{information source} is a stochastic process which generates a sequence of random variables $Y_{0}, Y_{1}, Y_{2}, \dots$, where each $Y_{i}$ takes values in a finite alphabet $\mathcal{Y}$. A \textbf{Markov source} is an information source whose underlying dynamics is given by a non-observable Markov chain, i.e., it is a hidden Markov model (HMM) \cite{Ephraim2002HiddenMP}. Note that in a Markov source, each random variable $Y_i$ is parameterized by the associated hidden state of the Markov model.

In the real world, after a Markov source generates a symbol, this symbol is fed into a communication channel to be delivered to a receiver. Formally, a discrete \textbf{communication channel} is a system consisting of an input alphabet $\mathcal{Y}$, output alphabet $\mathcal{R}$, and a transition matrix characterizing the probability masses $Q({r}\mid{y})$ of observing the output symbol $r \in \mathcal{R}$ given that the input symbol is $y \in \mathcal{Y}$. In this work, we consider \textbf{discrete memoryless channels}, where for any input string (sequence of symbols) ${\underline{y}}=\left(y_{1}, \ldots, y_{N}\right)$ and the output string ${\underline{r}}=\left(r_{1}, \ldots, r_{N}\right)$, the overall  probability transition matrix $Q({\underline{r}} \mid {\underline{y}})$ can be written as
$
Q(\underline{r} \mid {\underline{y}})=\prod_{i=1}^{N} Q\left(r_{i} \mid y_{i}\right)
$,
and each transition probability $Q\left(y_{i} \mid r_{i}\right)$ is independent of $i$.

\section{Inclusive formulation of DFAs}

\label{sec:inclusive_formulation_DFAs_and_information_sources}

\subsection{Decomposing the full state space of a DFA}
\label{sec:ground_inclusive_thermo}

Recall that the state space of a physical
computer implementing a particular logical computer is a Cartesian product of a set of accessible variables and a set
of inaccessible variables (or of multiple such sets of inaccessible variables, more generally).  This Cartesian product provides a coordinate system for the physical
computer, and the dynamics of probability distributions through those coordinates determines the thermodynamics of the physical computer. 

However, in
general this Cartesian product
coordinate system will not be
the coordinate system used directly
to specify the
dynamics of the logical computer
in its conventional, CS theory formulation.
As an example, suppose the physical computer implements a DFA. As illustrated below, we can design such a computer 
which contains the
SOI and a single bath. So this is a coordinate system
with two coordinates.
However, the update function of the DFA,
which determines the dynamics of the individual states of the full physical system, is directly provided in terms of
the set of triples $(s, z, \omega)$. That set of triples provides
a second coordinate system, 
which differs from the first one in general.

Accordingly, to analyze the thermodynamics of a given DFA and its update function, we need to specify a map
from the coordinate system of the logical computer
to the Cartesian product of accessible and inaccessible variables that comprise the physical computer.
We refer to such a map as a a \textbf{decomposition}
of the variables of the logical computer. So
for example, the map from the set of triples $(s, z, \omega)$
of a DFA to a Cartesian product of accessible and inaccessible variables is a \textbf{decomposition} of the set of such triples. 

In general, even for a fixed DFA, the
decomposition that is most appropriate will vary with different choices of the accessible and inaccessible variables, i.e., with different ways that an engineer could design the DFA.
Here we require that the state spaces of the
physical computer and the logical computer,
together with the decomposition
between the two, meet the following desiderata:
\begin{enumerate}
\item The full state space of the physical 
computer can be written as a Cartesian product $\mathrm{U} \bm{\times} \Pi_{i} \mathrm{V}^{i}$, where $\mathrm{U}$ is the state space of the SOI with states $u \in \mathrm{U}$, and $\mathrm{V}^{i}$ is the state space of the bath $i$ with states $v^{i} \in \mathrm{V}^{i}$.
\item The state of the SOI and those of the baths are statistically independent at the initial
time $t = 0$, i.e., $P_{0} = p_{0}  \bm{\times} \Pi_{i}\rho_{0}^{{i}}$, where $p_{0}$ is the initial distribution over the states of the SOI, and $\rho_{0}^{{i}}$ is that of the bath $i$.
\item For each bath, its initial distribution $\rho_{0}^{{i}}$ can be written as a Boltzmann distribution for an associated finite bath Hamiltonian, $B_i(v_{0}^{i} = v^{i})$, which does not change with time. So for all baths $i$,
we need the function
\eq{
B_i(v^i) := -\ln \rho_0^i(v^i)
}
to be finite-valued for all $v^i$
that can occur at \textit{any} iteration
(not just iteration $0$) with nonzero probability.
\item The decomposition map is injective for all
states of the logical computer that can occur with nonzero
probability.
\item The dynamics across the state space of the physical computer is 
deterministic and invertible for all of its states that 
occur with nonzero probability.
\end{enumerate}

Below we show that two particular decompositions 
satisfy these desiderata:
\begin{enumerate}[label=(\alph*)]
\item {\textit{Unilateral decomposition}: $\mathrm{U} =$ \{($s, z)$\} and $\mathrm{V} =$ \{$(\omega)$\}}, 
\item {\textit{Bilateral decomposition}: $\mathrm{U} =$ \{$(s, z, \omega[z])$\} and $\mathrm{V} =$ \{$( \omega[-z])$\}}
\end{enumerate}
The unilateral decomposition is appropriate for analyzing DFAs, in
the sense that the associated choice of accessible and inaccessible variables is
``reasonable" as a model of how an engineer would in fact be able to design a DFA. The presumption here is that the engineer can access the (physical variables specifying) the state of the DFA, $s$, and the pointer, $z$, but cannot do the same for the string $\omega$ that will be input to the DFA.

The bilateral decomposition is instead appropriate for analyzing Markov information sources and communication systems, again in the sense that a
real-world engineer will typically be able to build a device that directly manipulates the associated accessible variables, but not the associated inaccessible variables. In \cref{sec:multilateral_decomposition_dynamics}, we introduce a third decomposition, the \textit{multilateral decomposition}, which is an extension of the bilateral decomposition to scenarios with multiple baths. This last decomposition is
appropriate for analyzing the thermodynamics
of communication channels. 

In the next subsection we provide a fully formal definition of the dynamics of a DFA over a set of triples $(s, z, \omega)$. In
the following subsection, we describe how
to transform the dynamics from that
coordinate system of triples to the Cartesian product coordinate system $\mathrm{U} \bm{\times}_i \Pi_{i} \mathrm{V}^i$ (Until \cref{sec:multilateral_decomposition_dynamics}, we will concentrate on physical scenarios where there is only one bath coupled to the SOI, and
so for simplicity write $\mathrm{U}$ $\bm{\times}$ $\mathrm{V}$). We then use this description to satisfy our first desideratum. 
In the last subsection of this section, we discuss how to ensure that the remaining two desiderata are also satisfied.

We emphasize that while we consider transformations from one coordinate system to another for the case of DFAs, all of our analysis below would apply to many other computational systems which satisfy our desiderata. In general,
this requires paying special care to what variables in
the logical computer are exogeneous inputs, not defined
in the CS definition of the computer, and so need to be
identified with states of the bath(s) of the physical 
computer.

\subsection{The dynamics over the full DFA state space}
\label{sec:unilateral_decomposition_dynamics}

Here we will only consider maps from
the logical computer's state space to that
of the physical computer that obey the fourth
desideratum above. 
As a result,
the fifth desideratum holds iff
the dynamics across the state space of the logical
computer is 
deterministic and invertible for all of its states that 
occur with nonzero probability. 
Note though that considered as a function of just the DFA state $s$, the update function of a DFA need not be invertible in general, even if its dynamics is deterministic. 

This reflects the fact that in general, 
one cannot use a simple ``translation'' to go back and
forth between the variables in a physical computer 
and
those of the associated logical computer. Rather a
somewhat subtle coordinate transformation is required.
The next three subsections introduce a broad set of such
coordinate transformations, which will suffice for our
purposes in this paper.

To begin, note that in general 
there may be multiple pairs of a state of a DFA, $s$, and an input symbol, $y \in \omega$, which are all mapped by $f$ to the same next state of the DFA, $s'$. 
This would appear to violate our desideratum of invertible dynamics of the logical
computer. 
However, recall that the state space of the full system is the
space of all triples $(s, z, \omega)$. We need to ensure that the dynamics  over \textit{this} space 
is invertible, for all such triples that can occur with nonzero probability.
Fortunately, we can always do that, even if the update function applied to just the DFA's state is not invertible.

To do this, it is convenient to distinguish two kinds of triples, which we refer to as ``legal'' and ``illegal".
Formally, a triple $(s, z, \omega)$ is
\textbf{legal} if either of the following conditions is satisfied:
\begin{enumerate}
\item $z = 0, s = q_0$;
\item $(s, z, \omega) = (f(s', \omega[z]), z, \omega)$ for some legal triple, $(s', z - 1, \omega)$
\end{enumerate}
If a triple is not legal, it is called ``illegal".

We write $\Phi$ for the one-step iteration function for the dynamics of the full state space defined over
all triples $(s, z, \omega)$. $\Phi$ must be deterministic \textit{and} invertible.
Accordingly, for any $0 \leq z < \tau$, 
and any legal triple $(s, z, \omega$), we set $\Phi$ to 
\begin{equation}
\Phi:(s, z, \omega) \rightarrow(f(s, \omega[z]), z+1, \omega)
\label{eq:eqn_2}
\end{equation}
Hence $\Phi$ implements the DFA's update function, as desired, when it is run on a legal triple.

As the start state $q_0$ of the DFA is specified, given any current legal triple $(s, z, \omega)$, we know that the initial state was in fact $(q_0, 0, \omega)$. Therefore we can  recover the state for any value of the pointer uniquely, by evolving $(q_0, 0, \omega)$ forward for that number of iterations. In particular, we can recover the predecessor state of $(s, z, \omega)$ uniquely by evolving $(q_0, 0, \omega)$ forward
$z-1$ times. Thus, $\Phi^{-1}$
is a well-defined function for all
legal triples $(s, z, \omega)$, where $0 \leq z < \tau$.

Next, for any legal triple $(s, z = \tau, \omega)$, set
\begin{equation}
	\Phi(s,\tau,\omega) := (q_0, z, \omega)
\label{eq:eqn_3}
\end{equation}
This means that $(s, \tau, \omega)$ is a predecessor state of $(q_0, 0, \omega)$.
Note that since $\tau$ is fixed, we can evolve $(q_0, 0, \omega)$ forward $\tau$ iterations to uniquely obtain state $s$ in \cref{eq:eqn_2}. However, no legal states that are \textit{not} of the form $(s, \tau, \omega)$
get mapped by $\Phi$ to $(q_0, 0, \omega)$. Combining establishes that
this map from $(s, \tau, \omega)$ to $(q_0, 0, \omega)$
is also invertible. 

In addition, we stipulate that while
$\Phi$ can be stochastic, there is
probability zero of it mapping an illegal state into a legal state. (As an example,
in much of the analysis
below, for simplicity we take $\Phi$ to be
the identity
function when applied to any illegal state.)
Hence, to ensure that the update function of the DFA gets implemented in a deterministic invertible manner, we need to confirm that there is zero probability that $\Phi$ ever gets run
on an illegal triple. 
So long as we require that the initial distribution has support restricted to legal triples. So under
this requirement, the dynamics of the logical
computer is deterministic and invertible for all
states with nonzero probability, as required.

Note that this $\Phi$ is time-independent. However, as described in \cite{JarzynskiHamiltonian}, most of the calculations in the next section can be naturally extended to allow the dynamics to change with time. Furthermore,
for our preliminary analyses of the thermodynamics of DFAs we only need one bath, to specify the randomly chosen input string. 
However, for other analyses below, e.g. those involving communication channels,
we need to have more than one bath, 
since there are multiple statistically independent random processes affecting the SOI's evolution.

\subsection{Decompositions of DFAs}
\label{sec:form_coord}
We write the decomposition
from the space of triples $(s, z, \omega)$
to $\mathrm{U}  \bm{\times} \mathrm{V}$ as a bijective,
vector-valued function $g = (g^{U}, g^{V})$ with
domain $\mathcal{T} \times S \times  \Sigma^{\tau + 1}$.
The first component of this function gives the mapping to the state of the SOI,
and the second component gives the mapping to the state of the bath. We write this as
$g^U: \mathcal{T} \times S \times  \Sigma^{\tau + 1} \rightarrow \mathrm{U}$, and $g^V: \mathcal{T} \times S \times  \Sigma^{\tau + 1} \rightarrow \mathrm{V}$, respectively (with an obvious extension to
vector-valued functions $g$ whose image has more components when there are multiple baths). 

So we can
write the dynamics over the Cartesian product 
coordinate system of the physical computer as
\begin{equation}
\Phi^{comp}(u, v) = g(\Phi[g^{-1}(u, v)])
\label{eq:eqn_6}
\end{equation}
We will often abbreviate this as
$\Phi = g\Phi g^{-1}$, 
where the $\Phi$ on the RHS
is a function of $(s, z, \omega)$ while
the $\Phi$ on the LHS is a function of $(u, v)$.
We can use
\cref{eq:eqn_6} to specify how any distribution over $\mathrm{U} \bm{\times} \mathrm{V}$ evolves:
\begin{equation}
P_{t+1}(u, v) = P_t(\Phi^{-1}(u, v))
\label{eq:eqn_7}
\end{equation}

As an example, under the bilateral decomposition,
\begin{equation}
\begin{split}
g^{U}_{bi}(s, z, \omega) &:= (s, z, \omega[z]) \\
g^{V}_{bi}(s, z, \omega) &:= (\omega[-z])
\end{split}
\label{eq:eqn_4}
\end{equation}
with the bijection
\begin{equation}
g_{bi}(s, z, \omega) := (g^{U}_{bi}(s, z, \omega), g^{V}_{bi}(s, z, \omega))
\label{eq:eqn_5}
\end{equation}

Since the function $g$ is bijective, the fourth desideratum
is met. Together with the fact that the logical
computer's dynamics is deterministic and invertible,
as established in the previous subsection, this means
that the fifth desideratum is also met.

\subsection{Ensuring the second and third
desiderata are met}

The second desideratum is automatically
satisfied so long as we choose an appropriate initial distribution over the
set of triples $(s, z, \omega)$. For example, in the unilateral decomposition,
the initial distribution over the SOI states is a delta function over the $(s,z)$ pairs, equalling $1$ for $(s = q_0, z= 1)$. The initial bath state instead
specifies the input string to the DFA, and is given by sampling an appropriate distribution over such strings. In general, we impose no restrictions on that distribution, except that it be well-defined. In particular, we do not require that its support be restricted to strings that are in the language accepted by the DFA. 

Note that in the unilateral decomposition, due to the dynamics over triples $(s, z, \omega)$, the state of the bath never changes from its initial one. Accordingly, even if 
there are some values $v$ such that $\rho_0(v) = 0$, and so $B(v) = \infty$, 
those values can never occur at any iteration. So the third desideratum is automatically obeyed in the unilateral decomposition.

The time evolution of bath states in the bilateral decomposition is shown in \cref{fig:fig_h_4}.
Just as with the unilateral decomposition, the second desideratum is automatically
satisfied for the bilateral decomposition
so long as we choose an appropriate initial distribution over the
set of triples $(s, z, \omega)$.

There are some extra subtleties with establishing the third desideratum for the bilateral decomposition. Note
that the full string $\omega$ at $t = 0$
is generated as two separate substrings. The first
symbol, $\omega[0]$, is one component of the SOI, and so is generated by sampling the SOI. The string of the remaining symbols, $\omega[-0]$, is the full state of the bath
at $t=0$. So the full string $\omega$ is generated
at $t = 0$ by sampling $p_0(\omega[0])
\rho_0(\omega[-0])$. The associated
value of the bath Hamiltonian is
$B(\omega[-0]) = -\ln \rho_0(\omega[-0])$.

As the state of the bath at some later iteration
$z = t$ will be $\omega[-z]$, we must ensure that
its value under the function $B$ is nonzero
for any string $\omega$ that can occur with nonzero probability. In other words, it must
be the case that for all $\omega \in
\rm{supp} \left[p_0(\omega[0])
\rho_0(\omega[-0])\right]$, for all $t$,
$B(\omega[-z])$ is finite. This condition requires in turn that $p_0(\omega[0])\rho_0(\omega[-z])$ is 
also nonzero for that string $\omega$. In general, whether this condition is met will depend on the details of the update function of the DFA, as well as the precise distributions $p_0$ and $\rho_0$. However, 
this condition is always met if both 
$p_0$ and $\rho_0$ have full support, since
that ensures that every $\omega$ has nonzero probability of being generated at $t=0$. In such a case, the third desideratum is automatically satisfied for the bilateral decomposition.

\begin{figure}[h!]
\centering
        \includegraphics[width=0.5\linewidth]{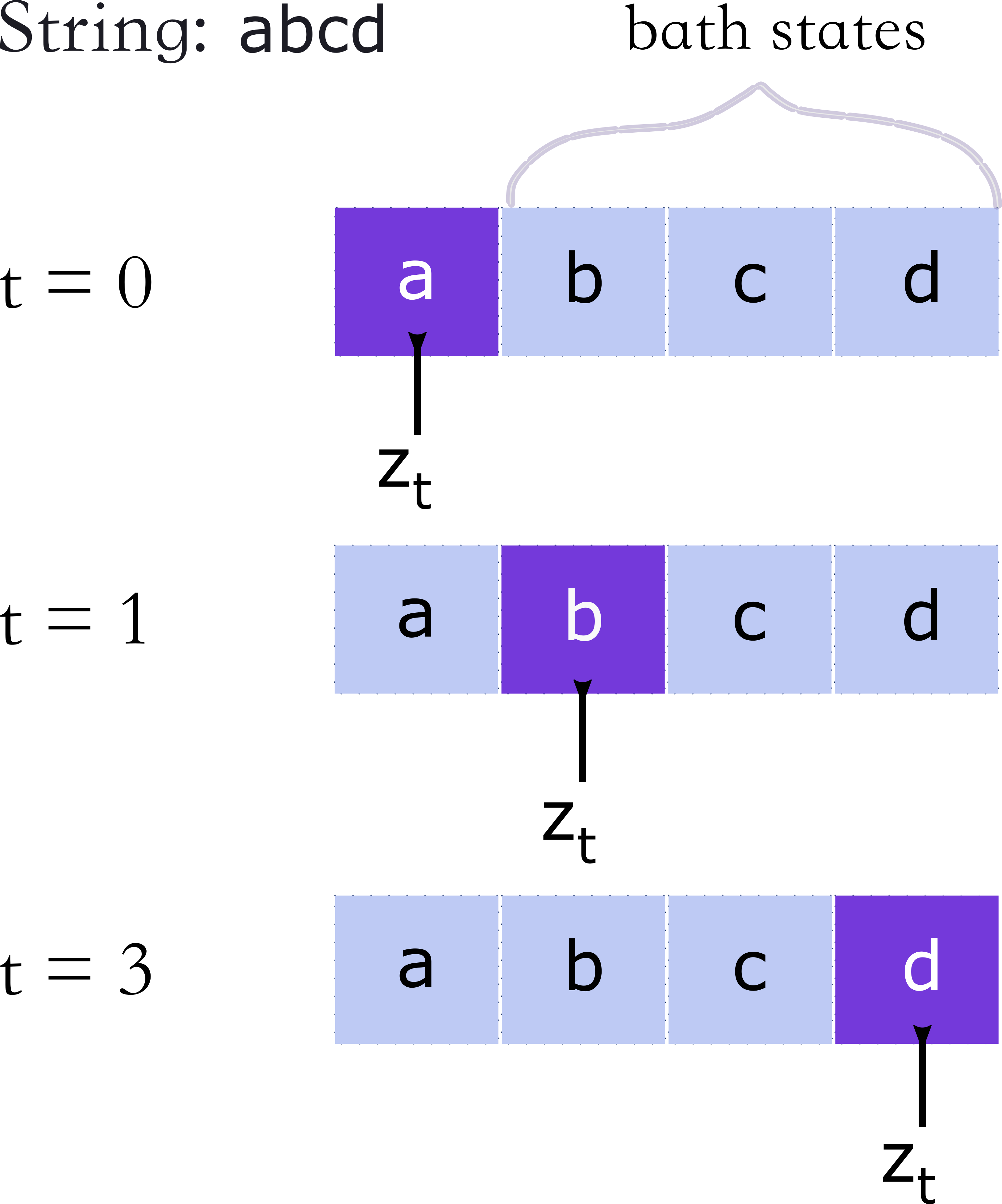}
   \caption{Recall that states of the baths are encoded as binary strings. In bilateral decomposition, in going from $t = 0$ to $t = 1$, the bath state changes from $\mathrm{bcd}$ to $\mathrm{acd}$.}
   \label{fig:fig_h_4}
\end{figure}

\section{Inclusive thermodynamics of DFAs}
\label{sec:thermo_inclusive}

In this section, we build on the inclusive formulation of the dynamics of physical computers which implement DFAs 
discussed just above, to introduce the inclusive thermodynamics of logical computers. 

First,
in \cref{sec:derive_ep}, we derive the expression for the REP, i.e., we derive a lower bound on the expected dissipated work during the reinitialization.  
Next, in \cref{sec:coding_efficiency} we show how to formulate Markov information sources (also known as HMMs or probabilistic DFAs), and in \cref{sec:thermodynamic_interpret_markov} we provide exact expressions for the expected EP incurred in the computational cycle of a Markov source. Building on these analyses in
\cref{sec:thermodynamic_complexity.information_source}, 
in the following subsections we present our IFT and mismatch cost result for the forward dynamics of logical computers which are modeled by the inclusive formulation. 
After that, in \cref{dfa_coord_baths} and \cref{sec:rate_distortion_channel} we extend the formulation of Markov sources to include multiple baths, and show how to interpret
such systems as models of communication channels. Next 
we consider systems with
multiple baths and derive a modified 
XFT which applies to not only our inclusive formulation of communication channels, but also to finite systems coupled to multiple finite baths as in Hamiltonian formulation.

Finally
in \cref{sec:thermo_complexity_lattice}, we review some classical results from set theory, relevant to DFAs. Then in \cref{sec:thermo_complexity_theorem}, we combine those results with our expression for EP in
\cref{sec:derive_ep}, 
to prove that for any regular language $L$, the DFA in $\Omega(L)$
which has minimal EP is the one with minimal size complexity. Our proof holds for all decompositions that satisfy the desiderata in \cref{sec:ground_inclusive_thermo}, for all dynamics, and for all iterations.

\subsection{Derivation of the expression for REP}
\label{sec:derive_ep}

In general, reinitizaliation of a physical computer involves a set of physically decoupled processes. The first of these processes reinitializes the accessible degrees of freedom of the SOI. The other processes reinitialize the separate baths. (Since we are for now focusing on the case of one bath, we will have a total of two processes.)

We do not make any assumptions for how the engineer reinitializes the accessible degrees of freedom. Because they can control these degrees of freedom, it's even conceivable that they could use the optimal thermodynamic protocol, involving a conventional, infinite external heat reservoir with $k_{B}T = 1$ to perform this reinitialization. Therefore we can only lower-bound the amount of work they need to expend in that reinitialization, by the amount that would be needed under that optimal thermodynamic protocol. Since we assume a uniform Hamiltonian over the states of the SOI, this minimal amount of work is given by the generalized Landauer bound \cite{ Esposito2010ThreeDF, DavidWolpertSTComputation}, as the change in the expected Shannon entropy of the SOI variables between the beginning and end of the reinitialization of the SOI:
\begin{equation}
\Delta H = [H(p_t(u)) - H(p_0(u))] 
\end{equation}

We know that the engineer needs to expend work in reinitializing the SOI, since the ending distribution of that reinitialization is a delta function. However, in general they might be able to \textit{recover} some work when they reinitialize the inaccessible degrees of freedom in the bath.
As the engineer has no direct control over how those inaccessible degrees of freedom are reinitialized, we 
adopt the ansatz that the amount of work which they can recover is \textit{upper}-bounded, by the
amount of heat that would flow
into a conventional, infinite
thermal reservoir with $k_{B}T = 1$, if the inaccessible degrees of freedom freely relaxed to the Boltzmann distribution of the Hamiltonian $B$ while coupled to that reservoir. Since energy is
conserved under free relaxation, this heat flow is just the change in the expected energy of the inaccessible degrees of freedom:
\begin{equation}
\bar{Q} = [\mathbb{E}_{\rho_0}(B(v)) - \mathbb{E}_{\rho_t}(B(v))]
\end{equation}

The dissipated work incurred in the reinitialization of the physical computer is given by the difference between the
amount of work spent on the SOI during that reinitialization, and the amount of work that can be recovered from the bath(s) during that reinitialization. A lower bound on this dissipated work is given by our lower bound on the amount of work that must be spent on the SOI, and the upper bound on the amount of work that can be extracted from the bath. Hence, the dissipated work is lower-bounded by
\begin{equation}
\overline{\sigma} = \Delta H - \bar{Q}
\label{eq:eqn_uation}
\end{equation}
As mentioned, we call $\overline{\sigma}$ the reinitialization EP (REP).
$\bar{Q}$ in \cref{eq:eqn_11}
is called the expected \textbf{entropy flow} (EF)
in the Hamiltonian framework. We will adopt the same terminology here. 

As an example, in the unilateral decomposition
of a DFA, since the state of the bath does not change
during the process, $\overline{Q} = 0$. Therefore
$\overline{\sigma} = \Delta H$. (In other words, the
EP in this case equals the generalized Landauer cost
of the SOI, i.e., the minimal \textit{EF} 
that could occur under the
CTMC-based version of ST.) So in particular, 
if the initial distribution of the SOI is
a delta function --- as is the case for example in a 
DFA --- then the expression for the EP simplifies further, to
$\overline{\sigma} = H(\rho_t)$.

\cref{eq:eqn_uation} was motivated
in \cite{massimiliano_lindenberg_vandenbroeck} through different physical considerations, identifying $\overline{\sigma}$ as the expected
entropy production of the forward process,
without any consideration of reinitialization.
(We refer the reader to \cref{sec:og_ep} 
for a short review of that argument in 
\cite{massimiliano_lindenberg_vandenbroeck}.)
The analysis in that paper
can be adapted to show that in the general case where there are $N$ baths, which are reinitialized separately of one another, in 
independent processes, \cref{eq:eqn_uation}
can be also written as
\begin{equation}
\overline{\sigma} = I(p_{t};\rho_{t}^{1}; \dots; \rho_{t}^{N}) + \sum_{i = 1}^{N} D\left[\rho_{t}^{i} \| \rho_{0}{^i}\right] \label{eq:eqn_11}
\end{equation} 
where the first term on RHS is the multi-information
between the SOI and all the baths at time $t$, and the second term is the sum of KL divergences between the initial time $t = 0$ distribution and the ending time distribution for each bath \cite{5392532}.
\subsection{Inclusive thermodynamics of Markov sources}
\label{sec:thermodynamic_complexity.information_source} 

In this section, we 
present some preliminary results concerning the thermodynamics of Markov information sources. We start in \cref{sec:coding_efficiency}
by 
showing how to formulate Markov sources in terms of DFAs. 
Subsequently in \cref{sec:thermodynamic_interpret_markov} we discuss the thermodynamic implications of modeling Markov sources as in \cref{sec:coding_efficiency}.
\subsubsection{Coding efficiency and Markov sources}

\label{sec:coding_efficiency}
 
Recall that in a standard DFA, at each
iteration $t$ the DFA is currently in some 
associated state $s_t$, and then receives
some symbol $\omega[t]$ to determine the next state
$s_{t+1}$, as specified by its update function $f$.

In the
standard interpretation of DFAs, we view each
such transition-specific symbol $\omega[t]$ as being \textit{received} by the DFA, and \textit{causing}
the DFA to implement that transition (due to the
update function of the DFA). We can
just as well view 
the exact same process as one in which the
DFA is currently in state $s_t$, and then \textit{generates} a symbol $y_{t} = \omega[t]$, as
it makes the (stochastic) transition from $s_t$
to $s_{t+1}$. 

Now restrict attention to the case where there is a conditional distribution $\pi(y | s)$ defined at every state of the DFA, and the  the 
strings $\omega$ of the DFA are generated by first sampling a distribution 
$\mathrm{P}(|\omega|)$ to generate a length of a string, and then running 
the process given in
Algorithm a total of $|\omega|$ times, to generate an associated sequence of symbols from $\Sigma$. This gives us a distribution $\mathrm{P}(\omega)$.

Next, recall from \cref{sec:preliminaries} that a Markov source is an HMM, which at each iteration $t$ generates a
symbol $y \in {\mathcal{Y}}$ by sampling 
a distribution that depends on the current
state of the hidden variable.
 Each successive symbol $y$ produced in this stochastic process specifies a state transition map over the set of hidden states, $\left(m\rightarrow m^{\prime} \right)$. Furthermore, any transition between
hidden states allowed by the HMM's adjacency
matrix, $\left(m \rightarrow m^{\prime} \right)$, 
 is specified by a unique symbol $y \in \mathcal{Y}$.

This establishes that we can view
the pseudo-code showing how to generate strings with Markov sources as an HMM, 
\begin{algorithm}[h!]
  \For{$i = 1;\ |\omega|;\ i++$}{
    Sample $\pi(y[i]\mid s[i-1])$ to get $y[i]$\;
    Set $s[i] = f(s[i-1], i-1, y[i])$ }
\end{algorithm}\\
where the states $s$
of the DFA are reinterpreted as hidden
states $m$ of an HMM, and we identify the set of symbols that the HMM can generate, $\cal{Y}$, as the alphabet
of the associated DFA $\Sigma$. 
Note that an important special case is where $\pi(y\mid s)$ is independent of $s$ for all iterations, so that the distribution over strings is given by IID sampling a fixed distribution. More precisely, when  $\omega=\left(\omega_{t}\right)_{t \geq 1}$ is a sequence of $\Sigma$-valued IID random variables, the successively visited states over the DFA obtained by processing $\omega$ gives a first-order homogeneous Markov chain \cite{Lladser2008MultiplePM}. In general though, $\pi(y|s)$ is dependent on $s$ and the Markov source induces an HMM over the DFA states.

Now suppose that the symbol $y_t$ generated at iteration $t$ of the HMM is first encoded according to a codebook, and then sent through a lossless channel to a receiver. Also assume that the receiver knows the mapping rules associating each hidden state transition with a unique symbol.  The receiver decodes the output of the channel to reconstruct $y_t$. If the receiver knows the previous hidden state $s_{t-1}$ at iteration $t$ 
(i.e., they know the previous
state of the DFA), then once they reconstruct $y_t$, they know the
current state, $s_t$, exactly. By induction, this means that
(assuming they know the initial state
of the DFA) they will always know what state $s_{t-1}$ was upon reconstructing $y_t$. 

Without loss of generality, assume that the codebook is a separate prefix-free code for
each state $s$ of the DFA.
Hence the codeword lengths for the code used in the DFA state $s$ must satisfy the Kraft inequality $\sum_{y} 2^{-l(y, s)} \leq 1$, where $l(y,s)$ gives the length of the associated prefix-free code for symbols $y$,
parameterized by the state $s$. The minimum expected codeword length $\bar{L}_{\min }(s)$ for each such $s$-parameterized code generated by the information source 
at each iteration $t$ satisfies \cite{Gray1990EntropyAI}:
\begin{equation}
H(Y_t \mid S_{t-1}) \leq \bar{L}^t_{\min }(S)<H(Y_t \mid S_{t-1})+1
\end{equation}
where
\begin{equation}
H(Y_{t} \mid S_{t-1})=-\sum_{s_{t-1} \in \mathcal{S}} \sum_{y_t \in \mathcal{Y}} P(s_{t-1}) \pi(y_t \mid s_{t-1}) \log \pi(y_t \mid s_{t-1})
\label{eqn:eqn_38}
\end{equation}
and $\pi(y | s)$ is the conditional distribution of symbols $y$
generated by the HMM at state $s$. 
\subsubsection{Thermodynamic interpretation of Markov sources}
\label{sec:thermodynamic_interpret_markov}
As we emphasized in \cref{sec:coding_efficiency}, for the bilateral decomposition of Markov sources there are two cases of interest: the simpler case considers $\pi(y | s)$ being independent of $s$ at all iterations, while the more general case considers $\pi(y | s)$ being dependent on $s$. Here, for both of those cases, we provide the exact expressions for the expected EP incurred in a computational cycle of a Markov source. 

First, we express the EP incurred in the time interval $[0, \tau]$ as
\begin{equation}
\begin{split}
\overline{\sigma} &= \Delta H(S, Y, Z) - \bar{Q} \\ &=  \Delta H(S, Y) + \Delta H(Z | S, Y) -\bar{Q} 
\end{split}
\label{eq:eqn__st21}
\end{equation}
Because $z$ is a deterministic variable, $\Delta H(Z | S, Y) = H_\tau(Z | S, Y)-  H_0(Z | S, Y)]$ is zero. In addition, since $s$ is fixed at $t = 0$, we can rewrite the above equation as
\begin{equation}
\begin{split}
\overline{\sigma} &= H_{\tau}(S, Y) - H_{0}(Y) - \bar{Q} \\
&=  H_{\tau}(S | Y) + H_{\tau}(Y) - H_0(Y) - \bar{Q}
\end{split}
\label{eq:eqn__st22.1}
\end{equation}
Similarly, we can write for the EP
\begin{equation}
\overline{\sigma} = H_{\tau}(Y | S)  + H_{\tau}(S)  - H_0(Y)- \bar{Q}
\label{eq:eqn__st22.2}
\end{equation} 
These are our first two results concerning the thermodynamics of DFA-based Markov sources.

Note that for any given DFA modeled under the unilateral decomposition, the three quantities $H_{\tau}(Y), H_{0}(Y),$ and $\bar{Q}$ only depend on the distribution over strings. In particular, they are independent of the number of states of the
DFA, or its update function. However, given that we define Markov sources in terms of $\pi(y | s)$, changes to the update 
function for a fixed $\pi(y | s)$ results in changes to all three of those quantities, in general.

Now, consider the first case we mentioned above, where $\pi(y | s)$ is chosen to be independent of $s$. Then those three quantities are independent
of all details of the DFA. So for this situation, we can easily compare any two DFAs, $A$ and $B$, based on their EP values incurred in a computational cycle using
the same $\pi(y)$. Using \cref{eq:eqn__st22.1}, the difference in their EPs is given by
\begin{equation}
H^{A}_{\tau}(S | Y) - H^{B}_{\tau}(S | Y)   
\end{equation}
This has a simple information-theoretic description: it is the difference in how
much information the ending symbol provides about the associated state of the hidden variable for the two associated HMMs.

In the second case considered above, regardless of $\pi$, using the equation for EP from \cref{eq:eqn__st22.2} we can write
\begin{equation}
\overline{\sigma} = \bar{L}^{\tau}_{\min }(S) + H_{\tau}(S) - \bar{Q}
\end{equation}
In particular, given any two HMMs generated with the same $\pi(y | s)$,
$A$ and $B$,
the difference in their EPs is
\begin{equation}
\bar{L}^{A}_{\min }(S) - \bar{L}^{B}_{\min }(S)
+ H_{\tau}^A(S) -  H_{\tau}^B(S)
\end{equation}
(where $\tau$ is  implicit.)
This gives a succinct relation between the thermodynamic cost of executing Markov sources with respect to a codebook:
In going from one HMM to another, the
EP changes by the sum of the associated change in minimal expected codeword length plus the associated change in the entropy of the ending distribution over hidden states.

\subsection{The integral fluctuation theorem for forward dynamics
of computational machines}
\label{sec:integral_ft_machines}
Recall that the state of the full computational system is written 
as $x \in X$, and we consider its values at integer-valued times,
$x_0, x_1, \ldots, x_\tau$. We write that entire trajectory
of $\tau + 1$ successive values of $x$ as $\vx$.
Formally, this is what we refer to as a \textbf{forward} trajectory,
to distinguish it from the dynamics when the computational
system gets reinitialized.
We will decompose that
trajectory into a trajectory of the SOI, $\vu$, together
with the trajectory of the bath, $\vv$. (When there are
multiple baths, we have multiple such bath trajectories, indicated as $\vv_i$.)

Since the Hamiltonian framework is a topic of
broad interest in the literature, in this section we derive an IFT constraining the
distribution of the values $\sigma(\vx)$, the EP of the forward process
that is the central concern
in the inclusive Hamiltonian framework.
While our approach applies more generally, we focus on
forward processes that implement the dynamics of
DFA.

Recall that the the main focus of the Hamiltonian
framework is the thermodynamics of
forward trajectories (where one also
imposes some assumptions concerning
the Hamiltonian of the full system that are not
necessary here).
Accordingly,
here we will use the terminology adopted in both the classical physics version of the Hamiltonian framework~\cite{JarzynskiHamiltonian,Strasberg2017StochasticTI,PhysRevResearch.2.033524}
as well as the quantum-mechanical version of the Hamiltonian framework
(sometimes referred to as the thermodynamics of ``open systems"~\cite{Funo2018QuantumFT,PhysRevX.8.031037}). 

To begin,
we define the \textbf{trajectory-level} EP as
\begin{equation}
\sigma(\vx) := \left[\ln p_0(u_0) - \ln p_\tau(u_\tau)\right]
    -  \left[B(v_0) - B(v_{\tau})\right]
\label{eq:eqn_1771}
\end{equation}
The expectation of $\sigma(\vx)$ over all trajectories
$\vx$ is just $\overline{\sigma}$, 
the expected EP during the reinitialization
process. Similarly, the expectation of the entropy flow $ Q = \left[B(v_0) - B(v_{\tau})\right]$ over all trajectories is the expected EF $\bar{Q}$ occurring in the reinitialization process (see \cref{eq:eqn_uation}). 

In addition though, assuming
an appropriate joint Hamiltonian over the state of the SOI
and the bath, the quantity on the RHS of \cref{eq:eqn_1771}
can be motivated without
any presumption of a reinitialization process, as the expected entropy
production that arises in the 
\textit{forward} process. 
(Indeed, showing that this quantity can
be identified with the expected EP 
in the forward process is one of the major
results of the Hamiltonian
framework~\cite{massimiliano_lindenberg_vandenbroeck, PhysRevLett.122.150603}.)

Note also that
the distribution of values $\sigma$ is the
distribution of values of forward process EP that would arise
in repeated sampling of the forward process. 
For the reasons just given, we know that those
two distributions over EP result in the same value
of expected EP. In general though, the 
distribution of values of EP generated during
the forward process differs
from the distribution of values of EP that would arise by reinitialization.

\subsubsection{Time-reversed processes}

\label{time_rev_IFT}

Here we introduce a few more concepts useful for deriving the IFT for the EP incurred during the forward computational process for a full system consisting of an SOI and a single bath. These concepts and our derivation can be easily generalized to include multiple baths.

Recall that for any decomposition of computational systems, where
$\mathrm{U}$ is the SOI and $\mathrm{V}$ is the bath, we stipulate that the initial distribution
over $\mathrm{U} \bm{\times} \mathrm{V}$ is a product distribution, and write 
$
{P}_{0}(s, z, \omega) = {p}_{0} {\rho}_{0}
$. 
Also recall that we require the support of ${P}_{0}$
to be restricted to legal triples in any decomposition.  Note that for many decompositions, the state of the bath $\mathrm{V}$ changes from one iteration $t$ to the next. Hence, as the SOI and the bath dynamically evolve, we must ensure that there is zero probability of $B(v_t)$ ever being infinite for accessible regions of the state space of the bath. (Recall the third desideratum in \cref{sec:ground_inclusive_thermo}.)

We define the reverse process in two steps. First,
the distribution over states \textit{of the bath} at the beginning of the reverse process is given by $\rho_0$, the same distribution that is set
at beginning of the forward process. After this distribution is 
sampled, the reverse process is
generated by running the forward process backwards in time. Using the formulation introduced in \cref{sec:inclusive_formulation_DFAs_and_information_sources}, this means that reverse processes over DFAs are obtained by iterating $\Phi^{-1}$
starting from the ending state of corresponding forward processes. 

For any forward trajectory $\vx$, we write the reverse of that trajectory as $\tvx$. The initial distribution of the reverse process and the dynamics of
the reverse process provides a distribution over 
the entire set of reverse trajectories, which we write
as $R(\tvx)$. We indicate marginalizations of that distribution in the usual way.
In particular, given a forward trajectory
$\vx$ going from $(q_0, z_0, \omega_0)$
to $(s_\tau, z_\tau, \omega_\tau)$ 
the probability $R(.)$ of the initial and final points in the associated
reversed trajectory $\tvx$ evolving from $\tvx_\tau$ to $\tvx_0$ is 
\eq{
R(\tvx_\tau, \tvx_0) = R( (s_\tau, z_\tau, \omega_\tau), (q_0, z_0, \omega_0))
\label{eq:5}
}
Using this notation, the conditional distribution
\eq{
	R(\tvx_0 = ( q_0, z_0, \omega_0) \,|\, \tvx_\tau = ( s_\tau, \tau, \omega_\tau))
	\label{eq:eqn_5.2}
}
is well-defined for all triples $( s_\tau, \tau, \omega_\tau)$. 
In addition, it equals $0$ for all illegal triples $( s_\tau, \tau, \omega_\tau)$ 
unless $(q_0, z_0,  \omega_0) = (s_\tau, \tau, \omega_\tau)$. 

To define the conditional distribution in \cref{eq:eqn_5.2} more formally, we
introduce $F(s, z, \omega)$ (resp. its inverse, $F^{-1}$) as the joint state of the full system, obtained by applying
$\Phi$ (resp. $\Phi^{-1}$) a total of $\tau$ times,
starting with $(s, z, \omega)$. 

Then the conditional distribution of the entire reverse trajectory,
given the entire forward trajectory, can be written as
\begin{equation}
R(\tvx_0 \,|\, \tvx_\tau) = 
    \delta(F^{-1}(\tvx_\tau), \tvx_0)
\end{equation}

Combining the above with the specification of the initial distribution over the full system in the reverse process, where we write $
R_{0}=p_{0} \rho_{0}$, we complete the definition of the reverse dynamics.
A subtle point is that mostly, for the initial distribution
of the reverse process to have this form of $R_0$, there must be nonzero probability of 
an initial state of the reverse process (an initial state for which $z = \tau$)
that is illegal. Those triples cannot possibly arise in
the forward process. As we run that reverse process, due to
our definition of $\Phi$, any such illegal point can only be mapped to other illegal points, as illustrated in \cref{fig:fig_trajectories}.

\begin{figure}[h!]
  \centering
   \includegraphics[width=0.8\linewidth]{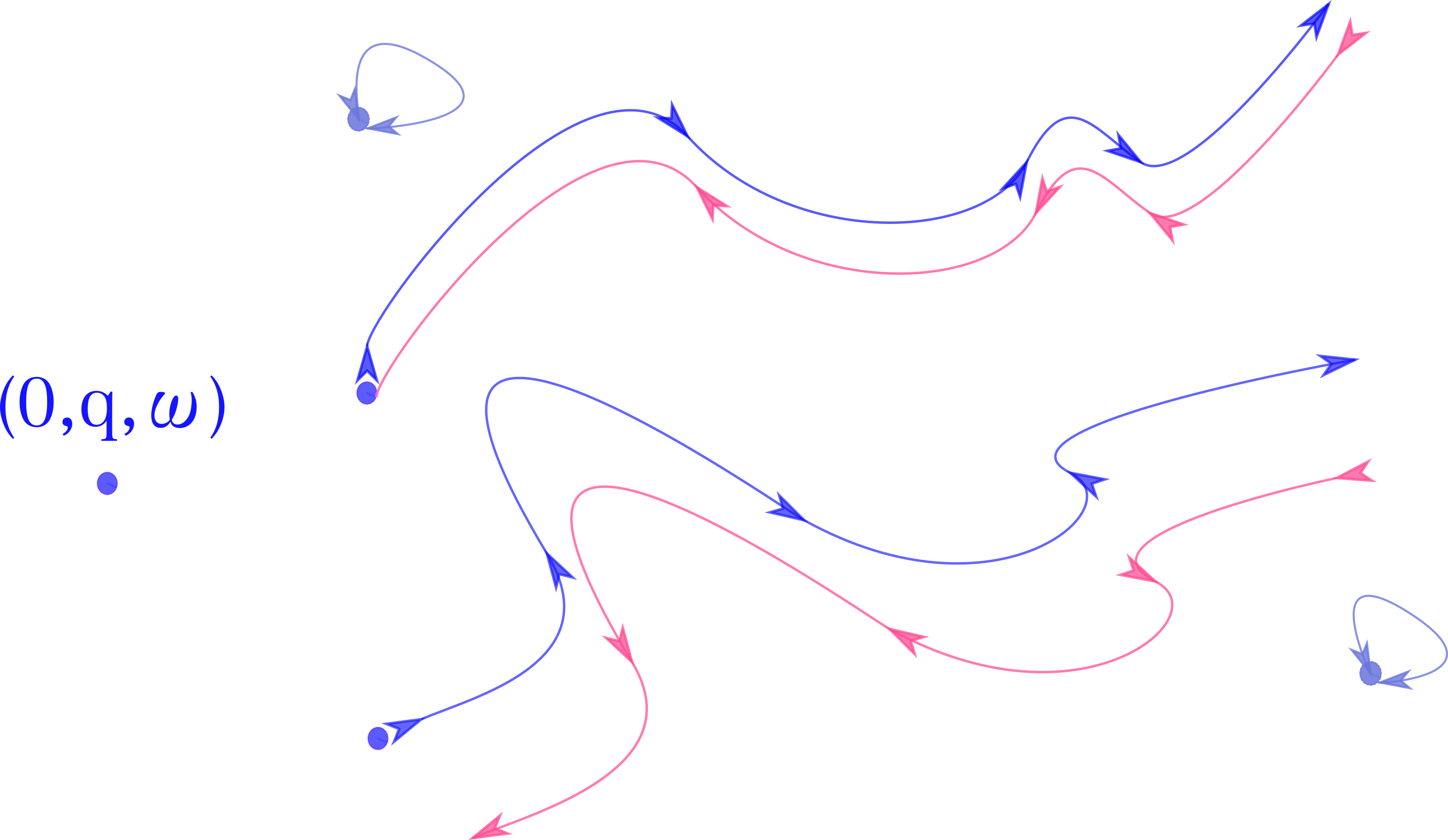} 
   \caption{Schematical illustration of forward (in blue) reverse (in orange) trajectories, which can be defined over legal triples $\{(s, z, \omega)\}$. Illegal triples can evolve under arbitrary dynamics so long as they are only mapped to illegal triples at each iteration. The indicated dynamics is the special case where every illegal state is mapped back to itself.}
   \label{fig:fig_trajectories}
\end{figure}

\subsubsection{Derivation}

Suppose that we
are provided a pair of (possibly vector-valued) functions $\mathrm{U}$ and $\mathrm{V}$,
both defined over the space $\mathcal{T} \times S \times \Sigma^\tau$,
such that the
function $g(., ., .)$ is a bijection.  As 
described above, the value
$u = u(s, z, \omega)$ is the state of the SOI and the value $v = v(s, z, \omega)$
is the state of the bath. We write the random variable of the trajectory of states of the SOI given
by the forward process as $\vu = (\vu_{0}, \ldots, \vu_{\tau})$. 
A particular value of that random variable is $(u_0, \ldots, u_\tau)$.
Similarly, we write the random variable of the trajectory of states of the SOI given
by the reverse process as $\tvu = (\tvu_0, \ldots, \tvu_\tau)$, where
a particular value of that random variable is 
$
(\tilde{u}_0, \ldots, \tilde{u}_\tau) = (u_\tau, \ldots, u_0)
$
(Note that $\tvu$ is a bijective function of ${u}$; they are not statistically independent variables
in any of the equations below.) We extend this notation in the obvious way
to the bath, defining the forward process random variable of the trajectory of the bath 
as $\vv = (\vv_0, \ldots, \vv_\tau)$ with values given by vectors $(v_0, \ldots, v_\tau)$,
and the associated reverse process trajectory as $\tvv = (\tvv_0, \ldots, \tvv_\tau)$ with
a particular value of that random variable 
written as the vector
$
(\tilde{v}_0, \ldots, \tilde{v}_\tau) = (v_\tau, \ldots, v_0)
$.

\label{sec:derivation_IFT}
Recall that the EF for any given forward trajectory is
$Q(v_0, v_\tau) := [B(v_0) - B(v_\tau)]$. Moreover, we can write
\begin{equation}
e^{Q(v_0, v_\tau)} = \frac{P\left(v_\tau\right)}{P\left(v_0\right)}
\label{eq:15}
\end{equation}
In addition,
\begin{equation}
Q(\tilde{v}_0, \tilde{v}_\tau) = Q({v}_\tau, {v}_0) = -Q({v}_0, {v}_\tau)
\label{eq:16}
\end{equation}
i.e., the EF generated by a reverse of a forward trajectory is the negative of
the EF generated under that forward trajectory. 

Given how the reverse
process is defined in terms of the forward process, with the forward trajectory being $(u_0, v_0, u_{\tau}, v_{\tau})$
and the reverse trajectory being $(u_{\tau}, v_{\tau}, u_0, v_0)$, 
\begin{equation}
\begin{split}
\frac{P(u_0, v_0, u_{\tau}, v_{\tau})}{R(u_{\tau}, v_{\tau}, u_0, v_0)} &= 
     \frac{P(u_{\tau}, v_{\tau} | u_0, v_0)P(u_0, v_0)}{R(u_0, v_0 | u_{\tau}, v_{\tau})R(u_{\tau}, v_{\tau})} \\
         &=
     \frac{P(u_{\tau}, v_{\tau} | u_0, v_0)P(u_0)P(v_0)}{R(u_0, v_0 | u_{\tau}, v_{\tau})R(u_{\tau})R (v_{\tau})}  \\
          &= \frac{P(u_{\tau}, v_{\tau} | u_0, v_0)P(u_0)P(v_0)}{R(u_0, v_0 | u_{\tau}, v_{\tau}){P}(u_{\tau})R (v_{\tau})}  \\
         &= \frac{e^{[\ln P_0(u_0)] - B(v_0)}}{e^{[\ln P(u_{\tau})] - B(v_{\tau})}} \\
         &= e^{\sigma(u_0, v_0, u_{\tau}, v_{\tau})}
\end{split}
\label{IFT}
\end{equation}
Next, as in the usual way of deriving an IFT from a DFT, we clear terms to write
\begin{equation}
e^{-\sigma(u_0, v_0, u_{\tau}, v_{\tau})} P(u_0, v_0, u_{\tau}, v_{\tau}) = R(u_{\tau}, v_{\tau}, u_0, v_0)
\label{eq:IFT_derivation}
\end{equation}
If we now integrate both sides \textit{over those quadruples $(u_0, v_0, u_{\tau}, v_{\tau})$ that
can occur in the forward process}, we get
\begin{equation}
\langle e^{\sigma} \rangle = 1 - \kappa
\label{eq:IFT_hamiltonian}
\end{equation}
In this equation,
$\kappa$ is the probability of generating a trajectory
$(u_{\tau}, v_{\tau}, u_0, v_0)$ under the \textit{reverse} process such that the \textit{forward} version of that trajectory, $(u_0, v_0, u_{\tau}, v_{\tau})$,
cannot be realized under the forward process, due to the restriction that the SOI start in its initialized state \footnote{
Note that if we were to integrate over \textit{all} quadruples $(u_0, v_0, u_{\tau}, v_{\tau})$,
that would include some for which the logarithm in the exponent in Eq.\,4 is infinite, so
that the integral would run over values for which the integrand is undefined.}. 
So the RHS of the IFT is not $1$ but rather less than $1$
because the support of the joint distribution at $t = 0$ is less than the full joint space.

\subsection{Mismatch costs for forward dynamics
of computational machines}
\label{sec:prior_costs_dfas}

Note that no matter how a given DFA is implemented as a physical
system, in general it can be initialized with an
arbitrary initial distribution $p_{0}$.
In this section, we analyze how changes to that initial distribution of a fixed physical system implementing a DFA affects the EP that is generated in any forward process of the computational cycle of the DFA \footnote{Recall that the computational cycle consist of two physical processes: one of them is the forward process where the DFA processes strings, and the other is the reinitialization. For reasons of space, we leave the extension of our analysis to reinizialiation process to future research.}.

Consider a physical system evolving from time
$t_0 = 0$ to $t_f = \tau$ in a fixed physical process. 
In general, as we vary the distribution 
over the states of the system at $t_0$,
we change the total EP generated in the
(fixed) physical process. Let $\alpha_{t_0}$ be the distribution 
which results in minimal EP of that fixed
physical process,
\begin{equation}
\alpha_{t_0}:=\underset{p_{t_0}}{\arg \min } \overline{\sigma} \left(p_{t_0}\right)
\end{equation}
Following earlier terminology in the 
literature, we refer to $\alpha_{t_0}$ as the \textbf{prior} distribution for the process. While $\alpha_{t_0}$ is an initial distribution that results in minimal EP, in general the 
process might begin with a different initial distribution $\beta_{t_0}$, resulting in a \textit{mismatch cost} $\mathcal{W}$, 
defined as the extra EP due to using a non-optimal
intial distribution,
$\mathcal{W} = \overline{\sigma}\left(\beta_{t_0}\right)-\overline{\sigma}\left(\alpha_{t_0}\right) \geqslant 0$ \cite{JStatMech, DavidWolpertSTComputation}. 

The main result of \cite{JStatMech} is that for a system evolving under a
CTMC from $t_0$ to $\tau$, 
$\mathcal{W}$ is the change from $t_0$ to $\tau$ of the
KL divergences between $\alpha_{t}$ and $\beta_t$:
\begin{equation}
\mathcal{W} =D\left(\beta_{t_0} \| \alpha_{t_0}\right)-D\left(\beta_{\tau} \| \alpha_{\tau}\right)
\label{eqn:dissipation_mis}
\end{equation}
\noindent
This result was extended in \cite{PhysRevE.104.054107, PhysRevResearch.2.033524} to apply to Langevin dynamics and to open quantum systems evolving in continuous time. In addition, 
fluctuation theorems for mismatch cost were derived
in those papers, as were 
differential formulas for the instantaneous dynamics of mismatch cost.

Recall that in \cite{PhysRevE.104.054107}
we calculated the
dissipated work incurred in the reinitialization 
process of the computational system
using the
CTMC-based approach to stochastic thermodynamics. 
Therefore we can use \cref{eqn:dissipation_mis}
to calculate the mismatch cost of that
reinitialization, getting
\eq{
\mathcal{W}_{reinit} = D\left(P_\tau(U, V) \| \alpha_\tau\right) - D\left(P_{t_0}(U, V) \| \alpha_{t_0}\right)
}
where $\alpha_\tau$ is the joint distribution over
$U \times V$ at time $\tau$ that would result
in minimal EP in the reinitialization, and
$\alpha_{t_0}$ is the form that distribution takes
after evolving according to the reinitialization
process described in \cref{PhysRevE.104.054107}.


However, we cannot apply  \cref{eqn:dissipation_mis}
to analyze the mismatch cost of the
\textit{forward} process, since unlike
the reinitialization process,
it is based on the
inclusive Hamiltonian approach, not the CTMC-based
approach. In this section we fill in this gap in the literature, by proving
that \cref{eqn:dissipation_mis} also applies
in an inclusive Hamiltonian setting, even
when the stopping time $\tau$ is a random variable \footnote{Recall that as mentioned above,
the formulation of open quantum thermodynamics involving partial traces \cite{Nielsen2000QuantumCA} is similar to the inclusive Hamiltonian framework.
Mismatch cost for that open quantum scenario is derived in \cite{PhysRevE.104.054107}. 
In particular, an appendix in that paper presents an  analysis that similar to the result derived here.}.

The derivations of \cref{eqn:dissipation_mis} in earlier analyses were based
on how the change in nonequilibrium free energy \textit{of the SOI} from the beginning to the end of the process gets modified if the initial distribution only is changed, with the thermodynamic process itself not changing.
These nonequilibrium free energies are defined as the difference between the entropy of the SOI and the expected energy of the SOI. In the inclusive Hamiltonian framework though, the EP is defined in terms of the
difference between the entropy of the SOI and the expected energy of the \textit{bath}, and so requires
a different analysis.

To begin this analysis, for simplicity we presume that $\alpha_{0}$ has full support, i.e., it is in the interior of the unit simplex \footnote{Note that
when analyzing mismatch cost we are considering changes to the
distribution over the state $u$ of the SOI only, not to the distribution over the full system. In
particular, our full support condition
only concerns the distribution over states
of the SOI. Moreover, in light of Appendix D of \cite{JStatMech} 
we can ensure that this full support condition is met if we can ensure that the distribution $P(u_\tau | u_0)$ has full support. 
In turn, one way to ensure that this
conditional distribution is by
appropriate choice of the update function of the DFA and the distribution over input strings. Another way is by introducing
appropriate stochasticity in the dynamics of illegal states.}. Under this assumption,
for any initial state distribution $\beta_{t_0}$, the directional derivative at $\alpha_{t_0}$ obeys \footnote{A more
general analysis would not need this assumption. That analysis relies on
defining ``islands'' and associated mathematical machinery \cite{Wolpert_2020}, and so we leave
it to future work.}
\begin{equation}
\left(\beta_{0}-\alpha_{0}\right) \cdot \nabla \overline{\sigma} \left(p_{0}\right)|_{p_0 = \alpha_0}=0 
\label{eq:eqn_239t7}
\end{equation}
where the gradient is with respect to the
components of the distribution $p_0$.

Note that the term
$\nabla \overline{\sigma} \left(p_{0}\right)|_{p_0} = \alpha_0$ in \cref{eq:eqn_239t7} 
is not restricted to lie in the unit simplex in general, i.e., it can point in a direction that results in probabilities that are not normalized.
Plugging in from \cref{eq:eqn_uation}, that
gradient is
\eq{
\nabla \overline{\sigma} \left(p_{0}\right)|_{p_0} =
\nabla \Delta H \left(p_{0}\right)|_{p_0}
  - \nabla \bar{Q} \left(p_{0}\right)|_{p_0}
  \label{eq:41aa}
}

Consider the second gradient in \cref{eq:41aa}.
Component $u_0$ of that gradient is
\eq{
\frac{\partial \bar{Q}(p_0, \rho_0)}{\partial p_0(u_0)}  &= \frac{\partial}{\partial p_0(u_0)} \left(\sum_v \rho_\tau(v)B(v)
    - \sum_v \rho_0(v)B(v)\right)
 \label{eq:41}
}
While $\rho_0(v)$ is independent of $p_0(u_0)$
for all states $v$, that
is not true of $\rho_t(v)$ in general, due to
interactions between the SOI and the bath
during the interval $[0, \tau]$. So
\cref{eq:41} reduces to
\eq{
\frac{\partial \bar{Q}(p_0, \rho_0)}{\partial p_0(u_0)} &= 
\frac{\partial}{\partial p_0(u_0)}
    \sum_{u_\tau,v_\tau} P_{\tau}(u_\tau, v_\tau)B(v_\tau) 
}
where $P_\tau(u_\tau, v_\tau)$ is the 
distribution at $t = \tau$ over the joint
states of the full system.
Since at $t = 0$ the joint
distribution is a product distribution,
we can expand $P_\tau(u_\tau, v_\tau)$ as
\eq{
p_0\left[F^{-1}_U(u_\tau, v_\tau)\right] \rho_0\left[F^{-1}_V(u_\tau, v_\tau)\right]
}
where as before, $F^{-1}(u_\tau, v_\tau)$ is the 
inverse dynamics function taking the final state $(u_\tau,
v_\tau)$ to the (unique) associated initial state,
$(u_0, v_0)$, with $F^{-1}_U(u_\tau, v_\tau)$  
being that associated $u_0$ and
$F^{-1}_V(u_\tau, v_\tau)$ being the associated $v_0$. Plugging in,
\begin{equation}
\begin{split}
\frac{\partial \bar{Q}(p_0, \rho_0)}{\partial p_0(u_0)} &= 
    \sum_{u_\tau,v_\tau} B(v_\tau)
    \rho_0\left[F^{-1}_V(u_\tau, v_\tau)\right]
\frac{\partial p_0\left[F^{-1}_U(u_\tau, v_\tau)\right]}{\partial p_0(u_0)} \\
&= 
    \sum_{u_\tau,v_\tau} B(v_\tau)
    \rho_0\left[F^{-1}_V(u_\tau, v_\tau)\right]
    \delta(u_0, F^{-1}_U(u_\tau, v_\tau)) \\
    &:= Q(u_0)
\end{split}
\label{eq:eqn_239t746}
\end{equation}

$Q(u_0)$ depends on $u_0$ but is independent 
of the associated value $p_0(u_0)$. Since this is 
true for all $u_0$, we can write $\bar{Q}(p_0, \rho_0)
= \sum_{u_0} p_0(u_0) Q(u_0)$, i.e., the change in
the expected heat of the bath 
is a linear function of the initial distribution
over states of the SOI.

We can combine \cref{eq:41aa} with
this result to evaluate the components of $\nabla \overline{\sigma} \left(p_{0}\right)$:
\begin{equation}
\begin{split}
\frac{\partial {\overline{\sigma}\left(p_{0}\right)} }{\partial p_0\left(u_{0}\right)} &=\left[-\sum_{u_{\tau}} p\left(u_{\tau} \mid u_{0}\right) \ln \left(\sum_{u_{0}^{\prime}} p_{0}\left(u_{0}^{\prime}\right) p\left(u_{\tau} \mid u_{0}^{\prime}\right)\right)-1\right]\\
&\qquad +\left[\ln p\left(u_{0}\right)+1\right]-Q\left(u_0)\right) \\
&=-\sum_{u_{\tau}} p\left(u_{\tau} \mid u_{0}\right) \ln p_{1}\left(u_{\tau} \right)+\ln p_{0}\left(v_{0}\right)- Q\left(u_{0}\right)
\end{split}
\label{eq:who_knows}
\end{equation}
Using equations \cref{eq:who_knows} and \cref{eq:eqn_uation}, we express the inner products in the following form
\begin{equation}
\begin{split}
\alpha_{0} \cdot \nabla \overline{\sigma}\left(\alpha_{0}\right) &=H\left(\alpha_{\tau}\right)-H\left(\alpha_{0}\right)-\bar{Q}=\overline{\sigma}\left(\alpha_{0}\right) \\
\beta_{0} \cdot \nabla \overline{\sigma} \left(\alpha_{0}\right) &=C\left(\beta_{\tau} \| \alpha_{\tau}\right)-C\left(\beta_{0}\|\alpha_{0}\right) -\bar{Q} \\
 & = D\left(\beta_{\tau} \| \alpha_{\tau}\right)-D\left(\beta_{0} \| \alpha_{0}\right)+\overline{\sigma}(\beta_{0})
\end{split}
\label{eq:eqn_239t8}
\end{equation}
where $C(p \| q):=-\sum_{x} p(x) \ln q(x)$ is the cross-entropy between two distributions. Combining with \cref{eq:eqn_239t7}, \cref{eq:eqn_239t8} gives \cref{eqn:dissipation_mis} for $\mathcal{W}$, as claimed. 

There are several comments worth making.
First, note that the same kind of reasoning used above
not only
gives the dependence of the EP on the initial distribution,
but also (for appropriately redefined priors $\alpha_{t_0}$) gives the dependence on the initial
distribution of the
change in entropy of the SOI; of the non-adiabatic EP;
and of the change in non-equilibrium free energy~\footnote{Formally, all these results follow by simply replacing the linear function $\overline{Q}(p_0)$ in \cref{eq:41aa} with some other linear function. For example, the formula for the dependence on the initial distribution of the change in entropy of the SOI is given by replacing
$\overline{Q}(p_0)$ with the ``linear'' function $0$.
See~\cite{PhysRevE.104.054107}.}. 

Second, when analyzing the thermodynamics of DFAs, 
it makes sense not only
to consider total EP generated up to a fixed time $\tau$, but also to consider 
total EP generated up to either $\tau$ or the earliest time that the DFA enters
an accept state, $\tau_{accept}(\bold{x})$, whichever comes first. In other words, 
when considering the thermodynamics of computational machines, we are often interested
in the thermodynamics up to a stopping time that is a random variable, set by the 
earliest instance when a particular stopping condition (e.g., an accept state being
reached) is met. This is true in particular for DFAs.

Formally, to accommodate such a random stopping time means that in defining the expected EP, we should average
over values $T(\bold{x}) := \min[\tau, \tau_{accept}(\bold{x})]$, as well as average over trajectories $\bold{x}$.
It turns out we can do this simply by fixing the update function of the DFA
to never again change the state of the DFA once it enters the accept state~\footnote{To 
confirm that this update function is invertible (as required by our
analysis), note that the pointer keeps changing \textit{its} value even after the 
DFA enters the accept state. This means that every legal state has a unique
legal predecessor state for this update function, as required.}. To see this,
note that since the input string doesn't change its state at \textit{any} times, 
the mutual information between the SOI and the bath won't change once the DFA enters the accept state.
So with this update function, the average of EP integrated up to the random time $T(\bold{x})$
is equivalent to the usual integrated EP, where we always integrate up to the fixed time $\tau$. 

\subsection{Inclusive formulation of communication systems} \label{sec:multilateral_decomposition_dynamics}
In this section, we apply the inclusive formulation of DFAs to communication channels. 
Specifically, we extend the formulation of the bilateral decomposition for Markov sources, so that in addition to a random process
representing an information source, there is a second random process representing the
noise in a channel that communicates the output of that information
source to a receiver. We start with \cref{dfa_coord_baths}, where we formalize this extension by introducing multiple baths (e.g., a second bath provides the
noise in that communication channel). We name the resulting decomposition which is suitable to model communication channels as the \textbf{multilateral decomposition}. In \cref{dfa_coord_baths}, we also show how to expand the formalism in \cref{sec:form_coord} to this decomposition.
In \cref{sec:rate_distortion_channel}, we present two thermodynamically motivated distortion functions for communication channels. In \cref{sec:xft_communication}, we exploit the properties of the multilateral decomposition to derive a modified version of XFTs.

\subsection{Multiple baths coordinate systems for communication channels}
\label{dfa_coord_baths}

We start by considering a Markov information source where as above, its states are written as $s$, and each state transition at iteration $t$ generates a unique associated output symbol $y$. Recall that $\omega$ is independent of $t$, being set by sampling an appropriate distribution $\mathrm{P}_0(\omega)$ before 
the SOI starts dynamically evolving. In the bilateral decomposition of Markov sources, we used one finite bath $\mathrm{B}_1$ to perturb the dynamics of the SOI, as Markov sources are only associated with one random process at a time (i.e., stochastic generation of symbols as the SOI evolves). We now extend this setup by including a communication channel which takes the input strings generated a Markov source, and provides outputs under noise. 

Note that communication channels are concerned with (at least) two distinct dynamical processes, first one being the generation of strings, and the second being the transfer of those strings via a noisy medium. To model communication channels, we ascribe this first process to the coupling of the SOI to bath $\mathrm{B}_1$, while we ascribe the latter to the coupling of the SOI to another bath to $\mathrm{B}_2$. 
The second bath manufactures the \textit{inherent} noise of a communication channel, in coherence with the usual communication theory formulation.

We write $\nu$ for a string of length $\tau$, which gives the state space of the second bath $\mathrm{B}_2$, sampled at the beginning of the computational cycle, just as $\omega$. As in Markov sources, we interpret the transition $s_t \rightarrow s_{t+1}$ as a generation of an input symbol $y_{t+1}$, which then goes through a channel $\mathrm{P}(r_{t+1} \mid y_{t+1})$, producing a channel output $r_{t+1} = (\omega[z_t] + \nu[z_t]) \, \mathrm{mod \, 2}.$ We assume that output strings of the channel are generated according to a distribution 
$\mathrm{P}(\rho \mid \omega)$, which is obtained by an IID sampling of the channel, where $\rho[z_{t+1}] = r_{t+1}$ for all $t$. 

This model of communication channels naturally leads to the following decomposition of the joint state space $\mathrm{U} \bm{\times} \Pi_{i} \mathrm{V}^{i}$, where we define the SOI state space over quadruples: $\mathrm{U}   =   \{s_t, \omega[z_t], (\omega[z_t] + \nu[z_t]) \, \mathrm{mod \, 2}, z_t\}$, and there are two distinct baths with state spaces $\mathrm{V}_1 =  \{\omega[-z_t]\}$, $\mathrm{V}_{2} = \{\nu[-z_t]\}$, respectively. We henceforth refer to this decomposition as the multilateral decomposition. In this decomposition, the set of elements of $\omega$ that
give the state of the bath $\mathrm{B}_1$ change from one iteration to the next, as well as the set of elements of $\nu$. Although, as in bilateral decomposition, we do \textit{not} change the definition of the bath from one iteration to the next in that decomposition. So in a similar way to \cref{sec:form_coord}, we provide below the bijective function $g$ for multilateral decomposition, which ensures the consistency of our approach.

Recall that $g$ maps any triple $(s, z, \omega)$ into an element of a fixed Cartesian product space $\mathrm{U} \bm{\times} \Pi_{i}\mathrm{V_{i}}$. For the multilateral decomposition, using the notation from \cref{sec:form_coord}, the space of the full system is given by $\mathcal{T} \times S \times \Sigma^{\tau + 1} \times \Sigma_{0, 1}^{\tau+1}$, where $\mathrm{U} := \mathcal{T} \times S \times  \Sigma \times \Sigma_{0, 1}$, and $\mathrm{V}_{1} := \Sigma^\tau,$ $\mathrm{V}_{2} = \Sigma_{0, 1}^{\tau}$. Here we define
three associated functions 
$g^{U}_{ml}: \mathcal{T} \times S \times \Sigma^{\tau + 1} \times \Sigma_{0, 1}^{\tau+1} \rightarrow U$, $g^{V_{1}}_{ml}: \mathcal{T} \times S \times \Sigma^{\tau + 1} \times \Sigma_{0, 1}^{\tau+1} \rightarrow \mathrm{V}_{1},$ and $g^{V_{2}}_{ml}: \mathcal{T} \times S \times \Sigma^{\tau + 1} \times \Sigma_{0, 1}^{\tau+1} \rightarrow \mathrm{V}_{2}$:
\eq{
g^{U}_{ml}(s, z, \omega, \nu) &:= (s, \omega[z], \omega[z] + \nu[z] \, \mathrm{mod \, 2}, z) \\
g^{V_{1}}_{ml}(s, z, \omega, \nu) &:= \omega[-z]
\\
g^{V_{2}}_{ml}(s, z, \omega, \nu) &:= \nu[-z]
}
which in turn allows us to express the dynamics over computational systems, including communication channels, which can be modeled by the multilateral decomposition.

\subsection{Thermodynamic rate-distortion functions}
\label{sec:rate_distortion_channel}
Assume that we have a discrete source that produces a sequence of symbols $Y_{1}, Y_{2}, \ldots, Y_{n}$ i.i.d. $\sim p(y), y \in \mathcal{Y}$, where each $Y_{k}$ assumes values in the source alphabet $\mathcal{Y}$. We consider the distortion that results when the source produces a symbol $y$, and an associated test communication channel delivers another symbol $r$ from the alphabet $\mathcal{R}$, as its representation of $y$. In general, it is possible to express such distortion as $ d\left(y_{k}, r_{k}\right)$, where $d(\cdot, \cdot): \mathcal{Y} \times \mathcal{R} \rightarrow[0, \infty)$ \cite{doi:https://doi.org/10.1002/047174882X.ch7}.

Recall from \cref{sec:preliminaries} that $Q=\{Q(r \mid y), y \in \mathcal{Y}, r \in \mathcal{R}\}$ formulates a communication channel. Given a source distribution $\{p(y)\}$, we associate with any such channel $Q$ two non-negative quantities $d(Q)$ and $I(Q)$ defined by
\begin{equation}
\begin{split}
&d(Q)=\sum_{y \in \mathcal{Y}} \sum_{r \in \mathcal{R}} p(y) Q(r \mid y) d(y, r) \\
&I(Q)=\sum_{y \in \mathcal{Y}} \sum_{r \in \mathcal{R}} p(y) Q(r \mid y) \log \left(\frac{Q(r \mid y)}{q(r)}\right)
\end{split}
\end{equation}
where
\begin{equation}
q(r)=\sum_{y \in \mathcal{Y}} p(y) Q(r \mid y).
\end{equation}

The quantities $d(Q)$ and $I(Q)$ are, respectively, the average distortion and the average Shannon mutual information associated with channel $Q$. The rate-distortion function of the i.i.d. source with symbol distribution $\left\{p(y)\right\}$ given the distortion constraint $d(\cdot, \cdot)$ is defined as the solution to the minimization problem:
\begin{equation}
R(D)=\min _{Q: d(Q) \leq D} I(Q) .
\end{equation}
which is solved over all conditional distributions $p(r \mid y)$ for which the joint distribution $p(y, r)=p(y) p(r \mid y)$ satisfies the expected distortion constraint. 

There are many possible distortion functions; which one to consider depends on the concerns of the engineer using the communication channel. We conjecture two possible, thermodynamically motivated distortion functions including EP below. 
\begin{equation}
\begin{split}
    d^{I}_{z} (y, r) & = \mathbb{E}(\sigma \mid r, z)-\mathbb{E}(\sigma \mid y, z) \\
    & = \sum_{y'}P(y' \mid r,z)\mathbb{E}(\sigma \mid r,y',z )\\
    & \qquad -\sum_{r'}P(r'\mid y,z)\mathbb{E}(\sigma \mid r',y,z)
\end{split}
\label{eq:distort_1}
\end{equation}
\begin{equation}
\begin{split}
    d^{II}_{z} (y, r) & = \mathbb{E}(\sigma \mid y, r, z)-\mathbb{E}(\sigma \mid y, z) 
\end{split}
\label{eq:distort_11}
\end{equation}

(Note that so long as the state spaces of the variables in \cref{eq:distort_1} and \cref{eq:distort_11}
are finite, we can always ensure
non-negativity of those two candidate distortion functions simply by subtracting their respective minimal values.)

\subsection{Modified XFTs for forward dynamics
of computational machines}
\label{sec:xft_communication}
In the Hamiltonian framework, an XFT is a symmetry relation concerning the heat exchange incurred during a thermodynamic process including multiple finite baths that are connected with one another, and that are initially prepared at different temperatures \cite{PhysRevLett.92.230602, PhysRevLett.123.090604}. The simplest form of an XFT, for two baths, is 
\begin{equation}
\ln \frac{p_{\tau}(+Q)}{p_{\tau}(-Q)}=\Delta \beta Q
\end{equation}
where $\Delta \beta =\beta_1 - \beta_2$ is the difference between the inverse temperatures at which the baths are prepared, and $Q$ is the heat flow between them during a thermodynamic process. 
These XFTs apply when there the baths are only connected to one another, without any SOI. 

In this subsection, we show how to derive a modified version of the XFTs for the inclusive framework, in which there an SOI coupled with multiple baths, e.g., in the inclusive framework formulation of communication channels. This XFT relates the work expended on the SOI during the forward process and the net heat flow among the baths during the forward process. (There is no simple variant of the standard inclusive Hamiltonian XFT for the reinitialization process, since that process involves coupling to infinite thermal reservoirs.)

We follow \cite{PhysRevLett.92.230602} in how we construct an XFT, considering the case of two baths for simplicity. (The generalization to scenarios involving more than two baths is straight-forward.) Write the inverse temperatures of the two baths, $\mathrm{B}_1$ and $\mathrm{B}_2$, as $\beta_1$ and $\beta_2$, respectively. We suppose 
that in the forward process the two baths are coupled to the SOI via an interaction term turned on at $t=0$ and turned off at $t=\tau$. 
We also assume that the work performed in switching all the interaction terms (both between the baths and between each bath and the SOI) on and off can be neglected. As in our derivation of an IFT, we will also consider a reverse process that is initialized to a Boltzmann distribution.

Next, recall our assumption that the initial and final Hamiltonians of the SOI are uniform. This means that the total change in the internal energy of the SOI during the forward process is zero, and the total work $W$ done on the SOI by any infinite external work reservoir changing the Hamiltonian of the SOI is simply the sum of the total heat flow into the baths, i.e., 
\eq{
W=\sum_{i} {Q}_i
\label{eq:55aaa}
}

Since both the forward process and the reverse process are initialized by sampling the respective Boltzmann distributions, the ratio of probabilities of a given forward trajectory and its reverse is 
\begin{equation}
 e^{\beta_{1} \Delta B_1 + \beta_{2} \Delta B_2}
\end{equation} 
where the changes are evaluated at respective values at the iterations $t = 0$ and $t = \tau$. Next, note that the two heat flow terms 
in \cref{eq:55aaa} are precisely $\Delta B_1$ and $\Delta B_2$, respectively.
Combining, we write the ratio of the probability of the forward trajectory to that of the reverse trajectory as
\begin{equation}
\frac{P(\vx_0)}{R(\tvx_0)} = e^{\beta_{1} Q_{1} + \beta_{2} Q_{2}} = e^{\beta_{1} Q + \beta_{2}[W - Q]} = e^{ Q \Delta \beta + \beta_2 W}
\label{eq:55}
\end{equation}
where $Q$ is shorthand for $Q_1$. Next, recall that by \cref{eq:16}, the heat transfer during a forward trajectory is the opposite of that during the reverse of that trajectory. Combining, we write
\begin{equation}
\begin{split}
p_{\tau}(Q) & = \int d \vx_0 P(\vx_0) \delta (Q- Q(\vx_0)) \\ 
& = e^{ Q \Delta \beta+ \beta_2 W} \int d \tvx_0 R(\tvx_0) \delta (Q+ Q(\tvx_0)) \\
& = e^{ Q \Delta \beta + \beta_2 W}p_{\tau}(-Q)
\end{split}
\label{eq:eqn_5143t5}
\end{equation}

Note that \cref{eq:eqn_5143t5} is in the form of the XFT derived by  Jarzynski and Wojcik, except here we have an SOI, whose Hamiltonian might change with time.

\subsection{Thermodynamic costs of computation and size complexity of DFAs}
 In this section, we prove that out of all DFAs which recognize the same language, the minimal complexity 
DFA (as defined in \cref{sec:preliminaries}) is the one with minimal EP at all iterations and for all dynamics.
\subsubsection{Additional relevant aspects of  DFAs}
\label{sec:thermo_complexity_lattice}

To proceed with our analysis, we need to first review some elementary concepts from set theory and the theory of automata.
Suppose $\alpha=\left\{A_{1}, A_{2}, \ldots\right\}$ and $\beta=\left\{B_{1}, B_{2}, \ldots\right\}$ are two partitions of $\Sigma^{*}$, with blocks $A_{i}$, and $B_{j}$, respectively. We write $\alpha \preceq \beta$ if each block $B_j$ of $\beta$ is a union of blocks of $\alpha$, and
in this case say that $\beta$ is a \textit{(partition) refinement} of $\alpha$. 
Given two partitions $\alpha$ and $\beta$, their \textit{meet}, $\alpha \wedge \beta$, (respectively, their \textit{join}, $\alpha \vee \beta$ ) is the coarsest (finest) partition which refines (is refined by) both $\alpha$ and $\beta$. The meet consists of all nonempty intersections of each block of $\alpha$ with each block of $\beta$. The join consists of the smallest subsets which are a union of blocks from both $\alpha$ and $\beta$,
\begin{equation}
\alpha \vee \beta=\left\{A_{i} \cap B_{j} \mid A_{i} \in \alpha, B_{j} \in \beta\right\}
\end{equation}
(Note that $\alpha \vee \beta = \beta$ if $\beta \preceq \alpha$.) 
Since any two blocks have such a meet and
a join, this combination of meets and joins is a {lattice}, which is known
as the partition lattice. 
 
Next, recall the definition of equivalence relation $\sim_{L}$ from \cref{sec:preliminaries}, which
provides
a partition over the strings in
$L$. 
Another equivalence relation, 
denoted by $\sim_{M(L)}$ (or just $\sim_M$
for short), is
defined over the same set as $\sim_L$,
and is specified by a particular DFA
$M$ that recognizes $L$.
Two strings $a, b$ are equivalent with respect to $\sim_{M}$ if processing them from the start state of $M$
puts $M$ in the same state:
\begin{equation}
a \sim_{M} b \Leftrightarrow f^{*}\left(q_{0}, a\right)=f^{*}\left(q_{0}, b\right)
\end{equation}
We write the partition lattice of
(the strings in) $L$ as 
$\Pi_{L}$ \cite{Birkhoff1950-BIRLT-2}.

We will sometimes refer to refinements of partitions as refinements of the associated equivalence relations. 
More precisely, we say that the equivalence relation $\sim$ is finer than the equivalence relation $\sim^{\prime}$ if
for any two $a, b \in \Sigma^{*}$
\begin{equation}
a  \sim b \Longrightarrow a \sim^{\prime} b
\end{equation}
\begin{figure}[h!]
  \centering
   \includegraphics[width=0.6\linewidth]{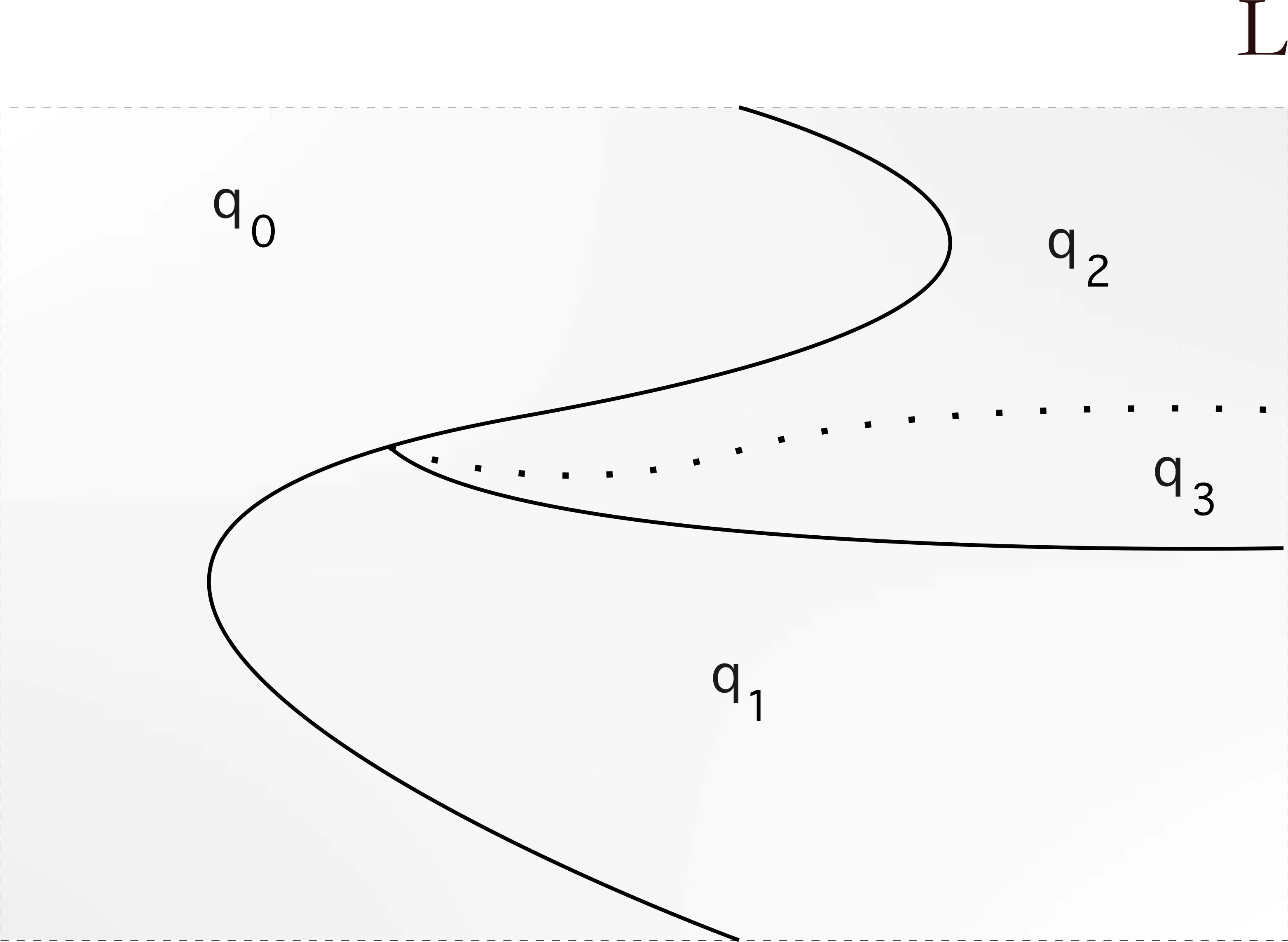}
   \caption{Bold lines distinguish the equivalence classes of the language $L = \{\omega \mid \omega \text{ is divisible by } 3\}$. $\sim_{L}$ introduces three equivalence classes over $L$, which can be bijectively mapped to the equivalence classes of $\sim_{M_{\text{min}}}$. A non-minimal DFA $M'$ which recognizes $L$ with a larger size, as in Fig. 2, is obtained by a coarser relation $\sim_{M'}$, whose equivalence classes are sketched by the bold and dashed lines combined.}
   \label{fig:fig_gh_5}
\end{figure}
   
Recall that by the MN theorem, for any language $L$, the equivalence classes of $\sim_{L}$ specify a unique minimal automaton $M_{\text{min}}$.
Moreover, the equivalence classes of $\sim_{M_{\text{min}}}$ (equivalently, the states of $M_{\text{min}}$) can be bijectively mapped to those of $\sim_L$, i.e., these equivalence classes are identical. In contrast, for any non-minimal DFA, $\sim_{M'}$ is always coarser than both $\sim_{M}$ and $\sim_{L}$ (we illustrate this by \cref{fig:fig_gh_5}). 

Let $\eta$ denote the partition of $\Sigma^{*}$ corresponding 
to the equivalence relation 
$\sim_{M_{\text{min}}}$ for some regular language $L$. Any refinement $\zeta$ of $\eta$ specifies some non-minimal DFA $M'$ which recognizes the same language $L$ that $M_{\text{min}}$ does. So the coarsest element of the lattice $\Pi_{L}$ defined above is given by the partition $\eta$. Finer elements of $\Pi_{L}$ correspond to non-minimal automata of different sizes, each formed by some refinement of $\sim_{M}$ or $\sim_L$.  
We write $H(p_{t}^{\eta})$ to express the entropy of the (probability distribution over the) equivalence classes of a minimal DFA $M_{\text{min}}$, and $H(p_{t}^{\zeta})$ to express the entropy of the (probability distribution over the) equivalence classes of a non-minimal DFA $M'$ equivalent to $M_{\text{min}}$.

\subsubsection{Minimal size complexity DFA is the one with minimal REP}
\label{sec:thermo_complexity_theorem}
In this subsection, we use the above to prove a relation between the size complexity of DFAs and the EP costs of executing DFAs.

\begin{lemma} Under the unilateral decomposition, 
at all iterations $t$, for any regular language $L$, the minimal DFA for $L$ has least EP of all DFAs that recognize $L$.
\label{lemma:1}
\end{lemma}

\begin{proof}
Without loss of generality, assume that any two DFAs in $\Omega(L)$ process the same set of strings sampled at time $t=0$.
Since there is a bijection between the set of the states of any DFA $M$ and the 
equivalence classes of $\sim_{M}$, the
entropy of the distribution over the states of the DFA $M$ at any iteration equals the entropy of the distribution over the equivalence class $\sim_M$ at that iteration.
So in particular, $H(p_{t}^{\zeta})$ is
the entropy at iteration $t$ of the distribution over the states of the non-minimal DFA $M'$, and $H(p_{t}^{\eta})$ is
the entropy at iteration $t$ of the distribution over the states of the minimal DFA $M_{\text{min}}$. 

Next, a classic result in ergodic theory is that for any distribution over a state space,
and any two partitions of that state space, $\alpha, \beta$, where $\beta$ is a refinement of $\alpha$,
$H(\beta) \geq H(\alpha)$, where $H(.)$ denotes the Shannon entropy of a distribution over a partition. (Proof follows by the chain rule of entropy for partitions \cite{Arnold1968ErgodicPO}.)  
So in particular, $H(p_{t}^{\zeta}) \geq H(p_{t}^{\eta})$. Combining establishes that at any iteration $t$
\begin{equation}
H_{M'} \geq H_{M_\text{min}}
\label{eq:eqn_partition_entropy}
\end{equation}
where we use $H_{M}$ as the shorthand notation for the entropy over the associated DFA states. 

Next, since the dynamics of $z$
is deterministic, 
\begin{equation}
\begin{split}
H(p_t(s, z)) &= H(p_t(s | z)) + H(p_t(z)) \\
&= H(p_t(s)) 
\end{split}
\end{equation}
for any fixed iteration $t$. So under
the unilateral decomposition, the change in entropy from iteration $0$ to $t$ of the (distribution over) the states of the SOI reduces to the associated change in the entropy of the (distribution over) the state of the DFA.

Moreover, because
all DFAs are initialized to a single state, $H_{t=0} = 0$, and so $\Delta H_{t = \tau} = H_{t = \tau}$. Using this, we write \cref{eq:eqn_partition_entropy} as
\begin{equation}
\Delta H_{M'} \geq \Delta  H_{M_\text{min}}
\label{eq:eqn_proof_lemma_1}
\end{equation}

In the unilateral decomposition, the lower bound on the dissipated work for any DFA is given by the change in the Shannon entropy over dynamical evolution. Hence, \cref{eq:eqn_proof_lemma_1} completes the proof.
\end{proof}

Lemma V.1 only concerns the unilateral decomposition. However, we can use it
to derive a general result that holds for all decompositions:

\begin{theorem} 
Under any decomposition, 
at all iterations $t$, for any regular language $L$, the minimal DFA for $L$ has least EP of all DFAs that recognize $L$.
\end{theorem}

\begin{proof}
Without loss of generality, assume that the bath dynamics are identical for any two DFAs in $\Omega(L)$. In other words, we consider any two equivalent DFAs that process the same arbitrary set of strings being generated under some fixed dynamics. This allows us compare the dissipation costs depending only on the size complexity.
Here for simplicity, we write the expected EP in the form of \cref{eq:eqn_11}, considering one bath (although our proof can be simply generalized to multiple baths)
\begin{equation}
\overline{\sigma}=I\left(p_{t} ; \rho_{t}\right)+D\left(\rho_{t} \| \rho_{0}\right)
\label{eq:eqn_110}
\end{equation}
Note that the second term on the RHS is the same for any two equivalent DFAs which process the same set of strings generated by identical physical processes. However, the first term on the RHS, which gives the mutual information between the SOI and the bath at any iteration $t$, might differ for equivalent minimal and non-minimal DFAs. We express this term for a minimal and a non-minimal DFA, respectively, as follows
\begin{equation}
\begin{split}
&I(p_{t}^{\zeta}; \rho_{t}) = H(\rho_{t})-H(\rho_{t} \mid p_{t}^{\zeta}) \\
&I(p_{t}^{\eta}; \rho_{t}) = H(\rho_{t})-H(\rho_{t} \mid p_{t}^{\eta})
\end{split}
\end{equation}
Since conditioning cannot increase entropy, $H(\rho_{t} \mid p_{t}^{\zeta}) \leq H(\rho_{t} \mid p_{t}^{\eta})$. Hence, $I(p_{t}^{\zeta}; \rho_{t}) \geq I(p_{t}^{\eta}; \rho_{t})$. Substituting to \cref{eq:eqn_110} completes the proof.
\end{proof}

Recall that in conventional, CTMC-based ST,
the EP given by a coarse-grained dynamics of a physical system is an upper bound
on the actual EP given by the fine-grained dynamics of that system. 
Thus, if all we’re interested in is achieving a certain coarse-grained dynamics, the information concerning the fine-grained dynamics 
within each coarse-grained bin (of the phase space where dynamical evolution takes place) is redundant. Such redundancy results in extra EP without
any compensating benefit for implementing the desired coarse-grained dynamics \cite{PhysRevE.87.041125}.

This line of reasoning does not apply directly in the inclusive framework, since
the physical system implementing a minimal DFA will not be a coarse-grained version of 
a physical system implementing a non-minimal DFA that recognizes the same language, in general. However Theorem V.2 provides an analogous result. 
The partition over the set of all strings processed by a non-minimal DFA is finer-grained than
it needs to be. Essentially, to avoid extra EP, it suffices to implement the same computation with a DFA whose partition would give a minimal number of bins, while retaining the needed information about the actual dynamics. For all regular languages, this is given by the minimal DFA. What Theorem V.2 shows is that the EP of that minimal DFA is the least that’s possible.

\section{Related literature}
\label{sec:remarks_dfas_paper}

There is a lot of previous research at the intersection of statistical physics and CS which
focuses on using the formal tools that have been developed for analyzing systems that are in a local thermal equilibrium and applying those tools to questions in CS \cite{Percus2006ComputationalCA}. One example of this work is the use of the replica method to analyze neural nets \cite{https://doi.org/10.48550/arxiv.2006.00256}. Another is the use of spin glass models and associated phenomena like phase transitions to investigate computational complexity questions \cite{10.5555/2086753, Percus2006ComputationalCA, MARTIN20013}.

This previous research does not concern the energetic costs of real physical systems that perform computation, which is the focus of the research presented in this paper.
There has already been some preliminary research on this issue.
Most of this work on the thermodynamics of computation thus far focused on TMs. Early work mainly addressed TMs evolving under deterministic and logically reversible dynamics. More recently, \cite{Kolchinsky2019ThermodynamicCO} provided an ST analysis of general-purpose TMs performing irreversible computation. However, it is well known that even for basic calculations it is not possible to provide an \textit{a priori} upper bound on the amount of tape a TM will use for any computational task \cite{10.5555/1540612}. (In fact, in part due to this, the CS research developed DFAs which can model real computers with finite resources \cite{5392601}).

It is possible to compare the analyses of 
the inclusive framework to the
standard ST analyses of computational machines,
e.g., to those of TMs. In these standard analyses, the EP is taken to be zero.  So all that can be calculated is the EF, which (under the usual uniform Hamiltonian assumption)
is given by the drop in entropy from the beginning to the end of the computational cycle.
That EF is equated with the minimal amount of work that
needs to be applied by an engineer, for a computational system to complete a computational cycle.

Furthermore, under this standard approach, it's not clear how to deal with
systems that get streams of stochastic inputs (like a DFA), as the computational cycle unfolds,
rather than getting all their
inputs at once, when the system is initialized
(like a TM does, as the initial string on its input tape) \footnote{Of course, one could
apply the same trick we do in our inclusive analysis of DFAs,
to re-express a stream of inputs as a single 
string of inputs that is determined when the system is initialized.
If one does that though, then the entropy drop of the full system
during a computational cycle is zero, and the analysis is rendered vacuous.}. Note as well that in this approach, if one re-initializes the
system after a run, all of the expended work is recovered, and the physically vacuous conclusion is that there is zero work, independent of the
details of the computational system.

As in the standard approach, we assume that a lower bound on the 
dissipated work
during the forward process is zero. Unlike the standard approach though,
we note that there is in general nonzero dissipated work 
in the re-initialization. This follows from the observation underlying our inclusive formulation, that there is always a difference between
accessible and inaccessible degrees of freedom. This
difference results in a strictly positive value for the minimal dissipated work
that occurs in a complete cycle of a computational system. It also allows
us to, e.g., derive a non-trivial IFT, whereas in the 
standard approaches with semi-static dynamics assumption, the 
distribution over EP values is a delta function about the value zero. As another generalization, we point out that prior work mostly assumed that a computational machine is implemented by a CTMC. Our inclusive framework allow the dynamics to be non-Markovian as well.

We propose that an inclusive thermodynamics of computation will offer rigorous tools to both advance theoretical research and engineer devices with enhanced performance. For instance, digital computers include hardware-implemented finite automata, and DFAs and NFAs are used to build systematic approaches for designing sequential circuits. The optimal design procedures of finite automata allow optimal or near-optimal physical circuit implementations \cite{10.5555/2500983}. Besides, in traditional architectures (CPU or GPU based), execution of an NFA may lead to a state space explosion leading to exponential run-time complexity. Hence, special-purpose architectures are required to run them \footnote{Micron’s Automata Processor (MAP) uses massively parallel in-memory processing capability of dynamic random-access memory for running the NFAs, so it can provide orders of magnitude performance improvement compared to traditional architectures \cite{6719386_MAP}. MAP is the first non-Von Neumann semiconductor device which can be programmed to execute thousands of NFAs in parallel to identify patterns in a data stream.}. We expect that the relations between the structural constraints, CS-based complexity, and thermodynamic bounds on computation will yield valuable design-inputs for next generation of the processors in computers \footnote{For instance, \cite{10.1145/3079079.3079100} corroborates our motive to particularly study the relations between energetics and size complexity, where an NFA partitioning algorithm that minimizes the number of state replications is introduced to maintain functionality --provided by a larger equivalent NFA-- with increased performance.}.

There might also be advantages for computers
run to analyze data sets in the more general natural sciences. For instance, modern biology research uses DFAs for the characterization of protein sequences through the identification of motifs present in them. High efficiency DFAs work as an accelerator for many applications in bioinformatics and data-mining \cite{protein_patterns_MAP}. It is well-known that the computational efficiency of DFAs is encoded in the descriptional complexity of DFAs, which we showed to be related to the thermodynamic costs of physically executing DFAs.

Finally, there have been some recent papers that considered non-equilibrium thermodynamics of finite automata, using a different approach from ours \cite{Finite_state_thermo_no, doi:10.1098/rsfs.2018.0037, https://doi.org/10.48550/arxiv.2208.06895}. 
In particular, \cite{https://doi.org/10.48550/arxiv.2208.06895} 
uses the CTMC version of stochastic thermodynamics, and
starts by recognizing that real-world (synchronous) physical
computers use a central clock to ensure that their dynamics
is periodic. As a result, the prior for the dynamics
at each successive iteration of
such a computer is the same. However, the distribution 
over states of the computer will vary from one iteration
to the next. Thus, there is unavoidable mismatch cost
in such a computer. \cite{https://doi.org/10.48550/arxiv.2208.06895}
analyzes this mismatch cost for DFAs, along with the ``local''
mismatch cost arising if the rate matrix governing the
dynamics of the DFA's state only has access to the current
input symbol rather than 
the full input string. In contrast, our
paper uses the inclusive Hamiltonian version of
stochastic thermodynamics, and considers the full EP,
not just mismatch cost. Also, we make 
no assumptions that the process is ``periodic'' 
or ``local'' in any sense.

In addition, \cite{PhysRevE.104.054107} derives a formula for expected mismatch cost for the classical setting and a non-infinite bath. That is very similar to
what we do here. However,~\cite{PhysRevE.104.054107} explicitly assumes that the SOI 
has a continuous-state state space, not a finite or countably infinite state space like the ones
considered in this paper. In addition, ~\cite{PhysRevE.104.054107} explicitly 
states that its results hold for the Hamiltonian of mean force formulation of EP, which differs from
the inclusive Hamiltonian formulation of EP considered in this paper~\footnote{In the 
Hamiltonian of mean force formulation, the initial joint distribution
of the SOI and the bath(s) is not a product distribution, in contrast to the case with the inclusive
Hamiltonian formulation considered in this paper. 
In addition, the precise definition of EP differs in
the two formulations, despite what is implied by the text above Eq. 20
of~\cite{Strasberg2017StochasticTI}.
(See Appendix A of~\cite{Strasberg2017StochasticTI}
for a more careful discussion confirming that the two formulations of EP differ)}. Finally, the results in~\cite{PhysRevE.104.054107} all assume that there is a fixed time $\tau$ 
at which the process ends, and considers all EP generated at times $t$ up to $\tau$.
It does not consider the case where there is a random variable giving the time
that a stopping condition is first satisfied, where the EP of interest is only
the EP up to that stopping time. The results in this paper also apply to the
case where the EP is integrated up to when a stopping condition is met. 

\section{Discussion}
\label{sec:discussion}

In this paper, we introduced a novel approach to quantify the thermodynamic costs of computation in systems operating far-from-equilibrium. In particular, we considered DFA-based computational machines, and derived a lower bound on the dissipated work that occurs in a cycle of any physical process that implements those machines. We showed that this lower bound is independent of the details of the physical process, and is formally identical to the irreversible EP arising in the Hamiltonian formulation of ST.

There are several possible directions for future work. One is to apply the inclusive framework to other computational machines in the Chomsky hierarchy, e.g., pushdown automata, or TMs. In particular, suppose
we have a conventional finite
cellular automaton (CA), whose initial state is
determined randomly, after which is evolves
in a deterministic process~\cite{Wolfram1983StatisticalMO}.
Such a system can be modeled as
a slight variant of the unilateral decomposition of a DFA, by identifying
each state of the DFA with a different state of the CA. For simplicity, 
suppose that the
CA state is a binary string of length $k$,
with the dynamics of that string given
by the CA evolution rule.
So the associated DFA has $2^k$ states.
We suppose as well the bath's state is a string of $k$ bits. At $t = -1$,
the string in the bath is generated by
random sampling. Then at $t=0$,
the state of the DFA is set to that
state of the bath, i.e., the bit
string in the bath is copied into
the state of the DFA / CA. After that
the DFA / CA undergoes deterministic dynamics, according to the associated
evolution rule of the CA. The analysis
above for the inclusive thermodynamics of
DFAs under the unilateral decomposition carries over almost directly, with minimal
change.

Another interesting direction for future work might explore more elaborate problems 
in CS. For instance, a way of building computational machines that recognize more complex regular languages is to combine simple languages by using Boolean operations, and construct DFAs that recognize the composite languages obtained from the elementary ones. It is an open question whether or not there are any concise relations between thermodynamic costs concerning such composite systems and the complexity of associated language operations. Along these lines, we also hope to investigate potential connections to generative grammars \cite{DeGiuli2019RandomLM} in the future work. Apart from that, it might be possible to translate our framework into the language of dynamical systems theory, vice versa, e.g., by considering the shift processes of DFAs \cite{Kitchens1997SymbolicDO, Delvenne2006DecidabilityAU}. All in all, we anticipate that there remains much more to understand about how our thermodynamic framework can be utilized for a further exploration of (the theory of) automata and formal languages, in a way that is reminiscent of von Neumann's notes \cite{Neumann1961JohnVN}.

Another open question is how the 
thermodynamics in the inclusive framework changes if we allow the length of the string input to the DFA to vary stochastically, and to then consider running a given DFA with a randomly set string length multiple times in succession \cite{PhysRevX.7.011019}.

Future
work might also involve extending our mismatch cost result to performing an analysis for the mismatch cost of the reinitialization process. (Note that the
reinitialization process ends with the 
distribution over states of the SOI being a delta function. Analyzing the mismatch cost for such a process requires the use of islands, in general.)
Analogously, the IFT derived in the main
text concerns the inclusive Hamiltonian definition of EP generated during the forward process. Some important future work would be to derive an IFT concerning the REP instead.

We conclude with two remarks. First, recall that reinitialization involves two parallel dynamic processes. The first is a quasi-statically slow change of a Hamiltonian of the SOI, so that the distribution of states of the SOI changes from $P_{\tau}(u_{\tau})$ back to its initialized form (which in the case of DFAs is a delta function). The SOI is in a Boltzmann distribution for its associated Hamiltonian throughout this process.

In contrast, the second process is a free relaxation of the bath under the Hamiltonian whose Boltzmann distribution is the initial distribution of the states of the bath.  That relaxation takes the distribution $P_\tau(v_\tau)$ to the initial distribution over bath states without any change to the Hamiltonian.  The distribution of values of the EP during the reinitialization is given by subtracting the random variable of the work required in the first process from the random variable of the work extracted in the second process. In general, that distribution will differ from the distribution of values of $\sigma$, calculated above. 
This is true despite the fact that the
expected EP in the forward process must
exactly cancel the expected EP in the 
reinitialization process, since the sequence of reinitialization after a forward process results in the exact same distribution over
states of the joint system. We refer the reader to \cref{sec:app_33} and \cref{sec:app_44} for more comments on this.

Finally, it is important to note that while DFAs and HMMs are closely related mathematically,
their inclusive thermodynamic properties are quite different. Mathematically, HMMs are somewhat related to what in computer science are called ``automata groups’’,
and are called ``information ratchets in \cite{DavidWolpertSTComputation}.
However again, this has no implications for how the inclusive thermodynamics of HMMs is related to 
analyses in the literature based on CTMCs for the thermodynamics of automata groups. So in particular,
our HMM results do not apply to the information
ratchets considered in detail in \cite{Boyd2017}.

\section{Acknowledgments}
David Wolpert was supported by the Santa Fe Institute and the 
National Science Foundation award CCF - 2221345. 
Gülce Kardeş thanks Matteo Marsili for stimulating discussions, and Pedro Harunari and Artemy Kolchinsky for very useful feedback on this manuscript.
Gülce Kardeş acknowledges support by the Quantitative Life Sciences, International Centre for Theoretical Physics.

\section{Appendices}
\appendix

\section{Expression for HEP}
\label{sec:og_ep}
As in \cite{massimiliano_lindenberg_vandenbroeck}, consider a finite setup composed of a finite SOI and finite bath(s).  Initially, at time $t = 0$, the system and the baths are decoupled, 
$P(0)=p(0) \prod_i \rho_i^{\mathrm{eq}}$. In turn,
the baths are initialized to  
\begin{equation}
    \rho_i(0)=\rho_i^{\mathrm{eq}}=\exp \left(-\beta_i \mathrm{H_i}\right)/ Z_i
\end{equation}
with $\mathrm{H_i}$ and $Z_i$ the corresponding bath Hamiltonian, and the partition function at $t=0$, respectively. The time evolution of all degrees of freedom in the joint system of the SOI and baths is governed by 
Hamiltonian (invertible) dynamics, with 
a Hamiltonian $\mathrm{H}(t)=\text{H}_{system}(t)+\sum_i \mathrm{H_i}+V(t)$. Here, $V(t)$ is an interaction term coupling the system and the baths, with the result
that considered in isolation, the SOI
evolves stochastically rather than deterministically.

\cite{massimiliano_lindenberg_vandenbroeck} 
analyzes this scenario. This analysis shows 
that the change in the Shannon entropy of the system from the initial time $t=0$
to any later time $t>0$
can be written as $\Delta H(t)=\Delta_i S(t)+\Delta_e S(t)$, with the conventional thermodynamic interpretation of
$\Delta_e S(t)$ as the change in expected energy of the
baths divided by the associated temperatures (see Eqn. 1--7 in \cite{massimiliano_lindenberg_vandenbroeck}). In our paper, we use a different notation ($\bar{\sigma}$ for $\Delta_i S(t)$ and $\bar{Q}$ for $\Delta_e S(t)$) and mainly consider classical systems. However, the decomposition of $\Delta H(t)$ in our paper is identical
to that in \cite{massimiliano_lindenberg_vandenbroeck}. Moreover, like the decomposition
in \cite{massimiliano_lindenberg_vandenbroeck}, our REP result applies to both classical and quantum systems.

\section{EP and EF for a Markov source processing semi-infinite strings}
\label{sec:app_1}
Here, we expand the EP and EF expressions presented in \cref{sec:inclusive_formulation_DFAs_and_information_sources} for Markov sources which process semi-infinite strings (hence the state space is countably infinite).

Recall that in the bilateral decomposition of Markov sources, we identify the finite set of triples \{$s, z, \omega[z]$\} as the SOI state space $\mathrm{U}$, while the set of all possible \{$\omega[-z]$\} correspond to the bath state space $\mathrm{V}$. The bath represents the data stream which provides the subsequent symbols to be processed by the SOI. As in the main text, we write $\omega[z = t] = y$. Notably
in the bilateral decomposition, the state of the bath $\omega[-z]$ changes in time, thus resulting in non-zero EF.

We formulate the expected energy of the bath at iteration $t \geq 1$ with $B(.)$ the Hamiltonian of the bath, and $\xi$ is a semi-infinite string tracked by the pointer,
\begin{equation}
\begin{split}
    \mathbb{E}(B_t) &= \sum_{\xi_{0:\infty}} P_{t} B_{t}(\xi_{0:t-1}, \xi_{t+1:\infty})  \\
       &= -\sum_{\xi_{0:\infty}} P_{t}
      \ln P(\omega_{1:t} = \xi_{0:t-1},\omega_{t+1:\infty} = \xi_{t+1:\infty})
\end{split}
\end{equation}
where $P_{t} = P_{t}(\omega_{0:t-1} = \xi_{0:t-1}, w_t = \xi_t, \omega_{t+1:\infty} = \xi_{t+1:\infty}).$ Since for general semi-infinite strings $\zeta_{1:\infty}$, we have
\begin{equation}
    B(\zeta) = -\ln P(\omega_{1:\infty} = \zeta_{1:\infty})
\end{equation}
EF for the bath at iteration $t = \tau$ can be written as
\begin{equation}
\bar{Q} =  \mathbb{E}(B_{t=0})-\mathbb{E}(B_{t=\tau})
\end{equation}
In certain scenarios we can analytically calculate EP and EF, given by these equations above. Consider for instance the Markov source which generates strings by emitting symbols drawn from a distribution $\pi(y \mid s)$, which is independent of $s$. So the change in the expected energy of the bath from $t=0$ to $t = n > 0$ is 
\eq
{\bar{Q} :&
    = \sum_{\xi_{1:n}} \ln [P(\omega_{1:n} = \xi_{1:n})] P(\omega_{1:n} = \xi_{1:n}) \\
    & \quad -\sum_{\xi_{1:n}} \ln [P(\omega_{1:n}  = \xi_{1:n})] P(\omega_{0:n-1} = \xi_{1:n}) \\
    & = - H(P(\omega_0)) - \sum_{i=1}^{n} H(P(\omega_i  \mid \omega_{i-1})) \\
    & \quad - \sum_{\xi_{1:n}} \ln [P(\omega_{1:n} = \xi_{1:n})] P(\omega_{0:n-1} = \xi_{1:n})
\label{eq:11111}
}
Similarly, since in any iteration $t$, $z = t$ with probability $1$, the change in entropy is
\eq{
\Delta H &:= H(P(S_n, Z_n, \omega_n)) - H(P(S_0, Z_0, \omega_0)) \\
   &= H(P(S_n, \omega_n)) - H(P(\omega_0)) \\
  &= H(P(\omega_n  \mid S_n)) + H(S_n) - H(\omega_0)
\label{eq:22222}    
}
EP is given by combining $\Delta H$ and $\bar{Q}$.
\section{Inclusive formulation of a DFA with bi-infinite strings}
\label{sec:app_2}
Here, we extend the formulation in \cref{sec:unilateral_decomposition_dynamics} to bi-infinite strings. In general, recall that the countably infinite state space of the full system is $(\{s, z, \omega\})$ where
$s \in S$ is a state of the DFA, and $\omega$ is a string. As mentioned in the main text, $\omega$ need not be finite but it can also be bi-infinite, $\omega^{\mathbb{Z}}=\left\{y=\left(y_{z}\right)_{z \in \mathbb{Z}}: y_{z} \in A \right.$ for all $\left.z \in \mathbb{Z}\right\}$. Each entry in $\omega$ is still a member of 
a finite alphabet $\Sigma$ that includes a special blank symbol.

Processing a bi-infinite string generates a bi-infinite trajectory over DFA states $\mathcal{S}^{\mathbb{Z}}=\left\{s=\left(s_{z}\right)_{z \in \mathbb{Z}}: s_{z} \in S \right.$ for all $\left.z \in \mathbb{Z}\right\}$. This allows us to extend the formulation in  \cref{sec:inclusive_formulation_DFAs_and_information_sources} by only slightly changing the definition of legal (resp. illegal) triples. 

For bi-infinite strings, we say that a triple $(s, z, \omega)$ is
\textbf{legal} if either of the following conditions is satisfied: $z \leq 0, s = q$, or $(s, z, \omega) = (z, f(s', \omega[z]), \omega)$ for some legal triple, $(z - 1, s', \omega)$. To ensure that the support of the system is restricted to legal triples, we impose that $p(z, s) = 1$ for  $(z, s) = (0, q)$, i.e., 
$P_{0}(s, z, \omega) = 0$ if either $z \neq 0$ 
and/or $s \ne q$.

\section{Interpretation of entropy production}
\label{sec:app_33}

Consider the inclusive Hamiltonian formula for EP
generated between initialization at $t_0$ and some arbitrary $t > t_0$: 
\eq{
\bar{\sigma}_{t_0, t} &= I_{t}(p_t ; \rho_t) - I_{t_0}(p_0; \rho_0) + D(\rho_t || \rho_0) \\
&= \Delta_{t_0, t}I +  D(\rho_{t} || \rho_{t_0})
}
The first term on the RHS in the second equation is additive over successive time intervals, i.e., for any pair of times $t_1< t_2$ both occurring after the initialization time $t_0$, it \textit{is} true that
\eq{
\Delta_{t_0, t_2}I = \Delta_{t_0, t_1}I + \Delta_{t_1, t_2}I 
}
So the first component of the formula for EP is additive. However, the second component is not:
\eq{
D(\rho_{t_2} || \rho_{t_0})\ne D(\rho_{t_2} || \rho_{t_1}) + D(\rho_{t_1} || \rho_{t_0})
}
Hence the entire formula for EP is not additive over
time,
\eq{
\sigma_{t_0, t_2} \ne \sigma_{t_0, t_1} + \sigma_{t_1, t_2}
}

That is in stark contrast to the meaning of EP as dissipated work in the CTMC-based approach to stochastic thermodynamics. It means that it is not correct to identify
$\sigma_{t_0, t_2} - \sigma_{t_0, t_1}$ as the ``EP" in going from $t_1$ to $t_2$.
Indeed, while the full EP 
from $t = t_0$  cannot be negative,
the change in that full EP as we go from $t_1$ to $t_2$ can be negative. 
This would arguably violate the second law
if in fact that change in EP were interpreted as conventional thermodynamic dissipated work. Another example of how problematic it would be to interpret $\sigma_{t_1, t_2}$
as the EP in going from $t_1$ to $t_2$ is that it's not clear what
``mismatch cost contribution to EP in going from iteration $t_1$ to $t_2$'' 
could mean. 

As a final illustration of this phenomenon, note that there is \textit{no} function whose only arguments are the distributions $P(u_{t_1}, v_{t_1})$
and $P(u_{t_2}, v_{t_2})$ that gives
the change in EP between $t_1$ and $t_2$, $\sigma_{t_0, t_2} - \sigma_{t_0, t_1}$. In this sense, change in full EP between those two times differs from the 
full EP of either of those times individually --- those two full EPs
\textit{can} both be calculated from just the associated two distributions. 
Now it would be possible to calculate $\sigma_{t_0, t_2} - \sigma_{t_0, t_1}$ with
a function that has those two distributions as arguments \textit{along with the bath's Hamiltonian function}. But the latter is not needed as an argument for calculating just EP from $t_0$ to $t$, $\sigma_{t_0, t}$, due to the assumption that the bath is in a Boltzmann distribution at $t_0$.

In summary, one 
should not interpret EP in the inclusive Hamiltonian framework as ``dissipated work", but as ``dissipated work \textit{relative to the distribution at $t_0$}. 
Viewed differently, the
EP at time $t_1$ in the inclusive Hamiltonian framework is the total work that has
been \textit{dissipated into the inaccessible degrees of freedom relative to the distribution at} $t_0$ and so cannot be directly recovered at $t_1$. 
Moreover, as mentioned above, this EP  --- this work that has been dissipated into the
inaccessible degrees of freedom --- can decrease between $t_1$ and $t_2 > t_1$.
This reflects the fact that while the engineer cannot access the inaccessible degrees of freedom directly at $t_1$ and recover work, in the time following $t_1$ some of the energy that was put into these inaccessible degrees of freedom
during $[t_0, t_1]$ would be transferred back into the accessible degrees of freedom.

\section{Relations among different approaches}
\label{sec:app_44}
One of the advantages of the inclusive thermodynamics framework is that by ignoring all details of
a physical process besides the computation it achieves, it helps focus on the thermodynamics
properties inherent in just that computation. As an example, there is still misunderstanding among
some researchers concerning the possible thermodynamic advantages of
``logically reversible Turing machines’’, or (logically reversible) ``Toffoli circuits’’~\cite{Landauer1961-nc}. These
systems are used to compute logically irreversible deterministic maps $f: x \in X \rightarrow f(x)$ without requiring
any work. Broadly speaking, to do this they replace the calculation of $f$ with the calculation of
$g : (x \in X, 0) \rightarrow (x, f(x))$, where the initial value $x$ is distributed according to some distribution over inputs, 
$p(x)$. Since $f$ is deterministic, the initial entropy of $(x, 0)$ equals the ending
entropy of $(x, f(x))$, and so the generalized Landauer cost implementing $g$ is zero. In contrast, if $f$ is logically
irreversible, then the entropy of $p(x)$ is greater than the entropy of $p(f(x))$, and so the generalized Landauer
cost of implementing $f$ is greater than $0$. This is interpreted to mean that implementing the reversible 
function $g$ is thermodynamically superior to implementing the irreversible function $f$.

Consider this scenario from the perspective of the inclusive framework however.
The system implementing $g$ has an extra accessible variable which is not accessible to
the system implementing $f$, a variable that must be initialized to $0$. (Formally, $f$ is a map over $X$ 
whereas $g$ is a map over $X^2$.) Therefore to ``compare apples to apples’’,
we should give the system implementing $f$ access to that extra coordinate which is
initialized to $0$, just like the system implementing $g$. Note though that the precise value
of that extra coordinate at the end of the computation is irrelevant when we implement 
the irreversible map $f$, in contrast to the case with $g$. So we can suppose that in implementing 
this variant of $f$, we use the extra, initialized coordinate 
as an information reservoir, thermalizing it as we run the computation. Formally, we replace $f$
with the stochastic map, $f^* : (x \in X, 0) \rightarrow (f(x), y)$, where $y$ is distributed according to 
an appropriate Boltzmann distribution, $P(y)$. In general, the Landauer cost of $f^*$ may
be greater than that of $g$, equal to it — or less than it. So in this more fair comparison, where
both systems being compared have access to the initialized coordinate, there is no \textit{a priori} advantage to using a reversible computer rather than an irreversible
one to implement $f$. 

\bibliography{main}

\end{document}